\documentclass{new_tlp}

\usepackage[ansinew]{inputenc}
\usepackage{lmodern}  
\usepackage[english]{babel}
\usepackage{mathtools}

\newcommand{\union}{\cup}
\newcommand{\inters}{\cap}
\newcommand\pfun{\mathrel{\ooalign{\hfil$\mapstochar\mkern5mu$\hfil\cr$\to$\cr}}}
\usepackage{wasysym} 
\newcommand{\TU}{\mathop{\mathbf{U}\mbox{}}\nolimits}
\newcommand{\TDiam}{\mathop{\Diamond\mbox{}}}
\newcommand{\TBox}{\mathop{\Box\mbox{}}}
\newenvironment{quoting}
{
\noindent\hspace{0.05\textwidth}
\begin{minipage}[c]{0.9\textwidth}
\vspace{2mm}
}
{
\vspace{2mm}
\end{minipage}
}

\newcommand{\calr}{\ensuremath{\mathcal{R}}}
\newcommand{\calt}{\ensuremath{\mathcal{T}}}
\newcommand{\cals}{\ensuremath{\mathcal{S}}}

\newcommand{\cala}{\ensuremath{\mathcal{A}}}

\newcommand{\cale}{\ensuremath{\mathcal{E}}}

\newtheorem{thm}{Theorem}
\newtheorem{prop}{Proposition}
\newtheorem{lemma}{Lemma}
\newtheorem{defn}{Definition}

\usepackage{tikz}
\usetikzlibrary{arrows, shapes.misc, decorations.pathmorphing}
\tikzset{sem/.style={-stealth, decorate, decoration={snake, amplitude=.4mm, segment length=4mm, post length=1mm}, line width=0.3mm}}
\tikzset{split/.style={-stealth, densely dashed, line width=0.3mm}}
\tikzset{sc/.style={-stealth, double, line width=0.1mm}}
\tikzset{>=stealth}
\tikzstyle{box} = [rectangle, draw, minimum width=5em, minimum height=2em]
\tikzstyle{state} = [rounded rectangle, draw, minimum size=2em]
\tikzstyle{trans} = [rectangle, draw, minimum size=1.9em]
\tikzstyle{pnstate} = [rounded rectangle, draw, minimum size=1.5em]
\tikzstyle{pntrans} = [rectangle, draw, thick, fill=black, minimum width=6mm, inner ysep=2pt]

\newcommand\rif{\;\;\mathrm{if}\;\;}
\newcommand\splitop{\mathop{\mathrm{split}}}
\newcommand\semop{\mathop{\mathrm{sem}}}
\newcommand\varsop{\mathop{\mathrm{vars}}}
\newcommand{\trans}[1]{\mathrel{\ooalign{$-$\cr\hidewidth\hbox{$\big[\mkern -2mu$}\cr}\hbox{$\mkern3mu$}#1\hbox{$\mkern-3mu$}\mathrel{\ooalign{$\big]$\cr\hidewidth\hbox{$\rightarrow\mkern-13mu$}\cr}\hbox{$\mkern14mu$}}}}

\title[Compositional specification in rewriting logic]
      {Compositional specification in rewriting logic%
      \thanks{Partially supported by MINECO Spanish project TRACES (TIN2015-67522-C3-3-R), and Comunidad de Madrid programs N-GREENS Software (S2013/ICE-2731) and BLOQUES (S2018/TCS-4339).}}

\author[\'O. Mart{\'\i}n, A. Verdejo \and N. Mart{\'\i}-Oliet]
       {\'OSCAR MART\'IN, ALBERTO VERDEJO and NARCISO MART\'I-OLIET
       \\ Facultad de Informática, Universidad Complutense de Madrid, Spain
       \\ \email{\{omartins,jalberto,narciso\}@ucm.es}}

\begin{document}

\maketitle

\noindent\textbf{Note:} This article has been published in \emph{Theory and Practice of Logic Programming}, 20(1), 44-98, \textcopyright\ Cambridge University Press.
\vspace{5mm}

\begin{abstract}
Rewriting logic is naturally concurrent: several subterms of the state term can be rewritten simultaneously. But state terms are global, which makes compositionality difficult to achieve. Compositionality here means being able to decompose a complex system into its functional components and code each as an isolated and encapsulated system. Our goal is to help bringing compositionality to system specification in rewriting logic. The base of our proposal is the operation that we call synchronous composition. We discuss the motivations and implications of our proposal, formalize it for rewriting logic and also for transition structures, to be used as semantics, and show the power of our approach with some examples.
\end{abstract}


\section{Introduction}

To anyone in the fields of computer science and engineering, the convenience of compositionality needs not be stressed. We choose the word \emph{compositionality}, instead of the weaker \emph{modularity}. The latter includes the mere separation of code in chunks for better organisation. Compositionality hints at the existence of interfaces and encapsulation like in, say, object-oriented programming. For example, as we show in Section~\ref{motivation-mutex}, compositionality allows a mutex controller to be a separate system, interacting, but not intersecting, with the controlled components. It also allows a complex system, like a computer architecture in the example in Section~\ref{motivation-comp-arch}, to be specified as a set of independent, interacting pieces.

Rewriting logic was started by Meseguer's paper~\cite{Meseguer1992}. Maude is the name of a language and a system based on rewriting logic. It is fully described in~\cite{ClavelDELMMT2007}. It has proven a useful tool for, among other things, the specification of non-deterministic and concurrent systems. In a rewriting logic specification, the state of the system is represented as an algebraic term. Concurrency is made possible by several subterms of the state term being rewritten simultaneously. Compositionality is difficult, because the state of the whole system is represented by a unique term at each moment. We intend to help fixing this limitation. Our main goal in this paper is not to introduce novel theoretical proposals about compositionality, but rather to adapt existing ideas and see what rewriting logic has to offer. We want to do so departing as little as possible from the setting of standard rewriting logic.

The language Maude includes a useful system of modules. When a Maude module is imported into another, the result is equivalent to copying an exact duplicate of the imported module. Thus, this kind of modularity helps organizing the code, but does not provide compositionality by itself. Asynchronous message passing is often used in Maude to maintain component specifications as separate as possible, but this also has its limitations (see Section~\ref{related} for a discussion).

Different works have very different views on compositionality (see, again, Section~\ref{related}). Separation of concerns is often enforced: the specification of the internal workings of a component is separated from the specification of its interaction with the environment. Although our work is only aimed at systems specified using rewriting logic, we have kept in mind the need to separate concerns. Indeed, our definition of composition does not require that components are specified in any particular formalism. We have called \emph{synchronous composition} the operation that allows assembling systems keeping them synchronised. It works in three phases: the internal workings of each component are specified using rewriting logic in a fairly standard way; then, properties are defined on the states and transitions of the system; synchronisation between components is then established by relating the values of properties from different components. Properties are functions that take values at states and transitions. They are defined extending, but not modifying the base system, so that the component needs not know whether it is going to be synchronised or used in isolation. This provides a high degree of separation of concerns. (We have drawn inspiration from the way Maude's model checker is used. See~\cite{ClavelDELMMT2007}.)

Although our main goal is not to introduce novel ideas, it is inevitable that adapting old ideas to new settings inspires or requires new developments. What we have called \emph{egalitarianism} consists in treating transitions as first-class citizens, the same as states. It is not a strict requirement for the developments described in this paper, but we find it convenient and have stuck to it. Also, we have found nowhere else our particular setting for composition, through synchronisation based on simple equality of potentially complex properties.

Using just equality between properties for composition may seem too simple, but it allows a wide range of possibilities. The logic for complex interaction can be put into new components, representing channels or connectors, to be specified in rewriting logic like the others. Asynchronous communication, in particular, can be emulated through the use of components implementing buffers and delays. See also Section~\ref{power} for more examples.

This is what the reader can expect from the rest of the paper. In Section~\ref{choices} we explain our motivations, our goals, and how they have driven us to particular choices and definitions. All the formalisms are in Section~\ref{formal-definitions}. Then, Section~\ref{simple-example} contains a simple but complete example for illustrative purposes. Section~\ref{related} discusses related work. The next steps we intend to complete are discussed in Section~\ref{future}, and conclusions are exposed in Section~\ref{conclusion}.

Our paper~\cite{MartinVM2018} contains additional examples, particularly a large one on the alternating bit protocol, aimed at showing the full power of our proposals. It also describes techniques we have found helpful when specifying systems using our method. Additional material can be found on our website: \texttt{http://maude.sip.ucm.es/syncprod}.

\section{Motivation, goals, and choices}
\label{choices}

We discuss the different choices we have been driven to make towards our definition of the synchronous composition operation, and the rationale behind each choice. We are interested in transition structures and rewrite systems, and the synchronous composition for both. We discuss them together.

One message we want to convey in this section is that most choices were imposed on us by two goals. Our number one goal, again, is to provide flexible tools for compositionality in rewriting logic. The second is to depart as little as possible from the standard definitions of rewriting logic, so that we can benefit from all the existing machinery around it: practical tools, like model checkers, state-space explorers, and execution engines; and theoretical results about the use of abstraction, semantics, and so on.

Rewriting logic has proven useful for tasks such as the formalisation of language semantics and the emulation of other logics. Our interest, however, is only focused in system modelling, specification, and analysis. This focus has also guided some of the choices described below.

As all this section is intended to be motivational, many technical details are only going to be properly explained and (hopefully) understood in later sections.

\subsection{First motivational example: mutual exclusion}
\label{motivation-mutex}

Think of a train, a very simple model of a train, that goes round a closed railway in which there is a station and a crossing with another railway. There are three points of interest in the railway, that we use as the states of our model. There are three transitions for moving between the three states.
\begin{center}
\begin{tikzpicture}[auto, xscale=0.36, yscale=0.27]
\draw (0, 0) circle (10);
\draw (0, 0) circle (11);
\draw [dotted] (-17, 0.5) -- (-16, 0.5);
\draw (-16, 0.5) -- (-9.5, 0.5);
\draw [dotted] (-9.5, 0.5) -- (-8.5, 0.5);
\draw [dotted] (-17, -0.5) -- (-16, -0.5);
\draw (-16, -0.5) -- (-9.5, -0.5);
\draw [dotted] (-9.5, -0.5) -- (-8.5, -0.5);
\node (train) at (0, 11.3) {\includegraphics[width=10mm]{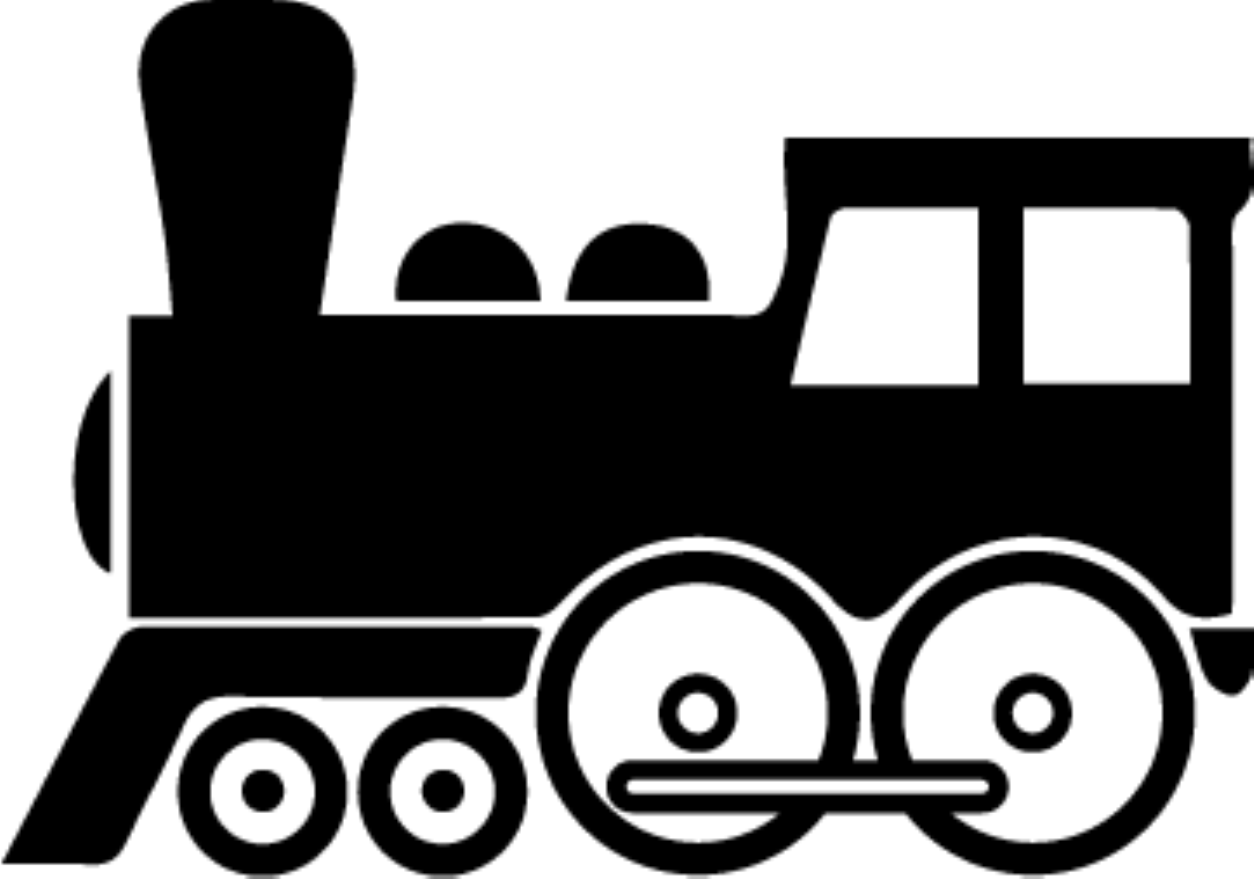}};
\node (train) at (-14.4, 0.8) {\scalebox{-1}[1]{\includegraphics[width=10mm]{train.pdf}}};
\node (station) at (10.4, 0.5) {\includegraphics[width=11mm]{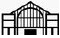}};

\node [state] (atStation) at (6, 0) {\texttt{atStation}};
\node [state] (beforeCrossing) at (-4, 3) {\texttt{beforeCrossing}};
\node [state] (afterCrossing) at (-4, -3) {\texttt{afterCrossing}};

\draw [->, bend right=35, out=-50] (atStation) edge node [sloped, above, yshift=-1mm] {\texttt{goingToCrossing}} (beforeCrossing);
\draw [->, bend right=35, out=-30, in=210] (beforeCrossing) edge node [left] {\texttt{crossing}} (afterCrossing);
\draw [->,  bend right=35, in=-120] (afterCrossing) edge node [below, sloped] {\texttt{goingToStation}} (atStation);
\end{tikzpicture}
\end{center}

In Maude-like notation:
\begin{verbatim}
rl [goingToCrossing] : atStation => beforeCrossing .
rl [crossing] : beforeCrossing => afterCrossing .
rl [goingToStation] : afterCrossing => atStation .
\end{verbatim}
The keyword \verb!rl! introduces a rewrite rule. The identifier in square brackets is the label of the rule. Rules describe transitions between states. To the left of the arrow (\verb"=>") is the origin state; to the right is the destination state.

Think, indeed, of two trains, both modelled the same, that share the piece of railway we have identified as \verb"crossing". Let us call these systems \verb"TRAIN1" and \verb"TRAIN2". This is a typical case where mutual exclusion is needed in the access to the crossing.

Modularity is desirable here. The two trains are independent systems that are better modelled separately and combined afterwards. Also, the control to ensure mutual exclusion can be exerted on the trains from outside. From a design point of view, the model of a train must specify how a train works, what actions it is able to perform, but not any control external to the machine. From a more practical point of view, having different concerns (the workings of the trains and the control) coded into different modules eases the tasks of specification and maintenance.

This is the module we propose to control for mutual exclusion:
\begin{center}
\begin{tikzpicture}[auto]
\node [state] (rem) at (0, 0) {\texttt{\;rem\;}};
\draw [<-] (rem.315) arc(-158:160:0.85);
\draw [<-] (rem.225) arc(158:-160:-0.85);
\node (c1) at (2.65, 0) {\texttt{crit(1)}};
\node (c2) at (-2.65, 0) {\texttt{crit(2)}};
\end{tikzpicture}
\end{center}
In Maude-like notation:
\begin{verbatim}
rl [crit(I)] : rem => rem .
\end{verbatim}
This system, that we call \verb!MUTEX!, is composed with the trains to produce a single system that we denote as $\texttt{TRAIN1} \,\|\, \texttt{TRAIN2} \,\|\, \texttt{MUTEX}$. For this composed system to work properly, we need to make sure that the transition \verb!crossing! in system \verb!TRAIN1! takes place synchronised with (that is, simultaneously to) transition \verb!crit(1)! in system \verb!MUTEX!; and \verb!crossing! in \verb!TRAIN2! with \verb!crit(2)!. Certainly, transitions \verb!crit(1)! and \verb!crit(2)! cannot happen at the same time and, thus, mutual exclusion is ensured. This module \verb!MUTEX! can be used to ensure mutual exclusion on any two systems, with appropriate synchronisation criteria.

This opens the door to compositional verification (that only gets some passing mentions in the present paper). The system \verb!MUTEX! satisfies mutual exclusion, represented by the linear temporal logic (LTL) formula
\[\TBox(\neg\texttt{crit(1)}\lor\neg\texttt{crit(2)}),\]
where \verb!crit(I)! is used as a proposition that holds true when the respective transition is executing, and false otherwise. Therefore, we can also assert that the composed, synchronised system satisfies mutual exclusion in the form
\[\TBox(\neg\texttt{crossing}_1\lor\neg\texttt{crossing}_2),\]
where $\texttt{crossing}_1$ represents a proposition that holds true when transition \verb!crossing! is executing in system \verb!TRAIN1!, and false otherwise; and respectively for $\texttt{crossing}_2$.

But there are issues to solve. We have used terms with variables as rule labels (like \verb!crit(I)!), which is not a standard feature of rewriting logic. Also, the critical section of a system may consist of several consecutive transitions, and all of them need to be synchronised with the same \verb!crit(I)! transition in \verb!MUTEX!. All this is explored and solved below.

\subsection{Second motivational example: computer architecture}
\label{motivation-comp-arch}

We introduce now a slightly more complex example. It models a schematic computer architecture in which a processor works with an external memory for data storage and a separate memory to store the program to be executed. We can picture the complete system like this:
\begin{center}
\begin{tikzpicture}[auto]
\node [box] (program) at (0, 0) {Program};
\node [box] (processor) at (4.5, 0) {Processor};
\node [box] (memory) at (9, 0) {Memory};
\draw [->] (program) edge node {instruction bus} (processor);
\draw [->, bend left] (processor) edge node {address bus} (memory);
\draw [<->, bend right, below] (processor) edge node {data bus} (memory);
\end{tikzpicture}
\end{center}
The processor component also includes a register for the instruction currently being executed, and another register for the last piece of data read from memory or to be written to memory.

In object-oriented and similar methodologies, the three entities---program, processor, and memory---would be coded as independent modules (classes), with internal details hidden, and with the needed operations in the interfaces to allow for the exchange of requests and data.

A typical rewriting logic specification could go like this:
\begin{itemize}
\item the stored program is represented by a set of instructions, each one similar to \verb!(1, w, 3, 7)!, meaning: ``instruction number 1 asks to write at memory address 3 the data 7'';
\item the memory is represented as a set of pairs (address, data);
\item the processor stores the counter for the next instruction to be executed, plus an instruction like \verb!(w, 3, 7)! in its instruction register, plus a piece of data in its data register.
\end{itemize}
The steps in the evolution of such a system would be represented by rewriting rules like this one, that executes a writing instruction already stored in the processor's register:
\begin{verbatim}
--- Part of the COMPUTER specification
   rl [execW] : Program: SomeSetOfInstructions
                Processor:
                   ProgCounter: N
                   Instr: (w, A, D)
                   Data: D'
                Memory: (A, D'') RestOfMemory
             => Program: SomeSetOfInstructions
                Processor:
                   ProgCounter: N + 1
                   Instr: void
                   Data: D
                Memory: (A, D) RestOfMemory .
\end{verbatim}
Maude allows using flexible syntax; all the elements in that rule that are not variables are added syntax. The rule, put in words, is saying: ``If the instruction just read by the processor is a writing requirement for address \verb"A" and data \verb"D" (that is, \verb"(w, A, D)"), then overwrite whichever data is currently stored in the memory at address \verb"A" with the new data \verb"D", update also the data register with \verb"D", and add one to the program counter.'' Other rules would take care of other instructions.

This is simple and useful for some goals. But it is not modular. The problem is that the rule above involves two components---the processor and the memory---and there is no way, in the setting of traditional rewriting logic, to model their behaviour independently. This is what we are after: being able to specify three separate systems for the three components of the computer, and to make them evolve in a synchronised way. In the processor, there would be a rule like this:
\begin{verbatim}
--- In module PROCESSOR
   rl [doingW] : ProgCounter: N
                 Instr: (w, A, D)
                 Data: D'
              => ProgCounter: N + 1
                 Instr: void
                 Data: D .
\end{verbatim}
Then, in the memory, there would be:
\begin{verbatim}
--- In module MEMORY
   rl [updating] : (A, D'') RestOfMemory
                => (A, D) RestOfMemory .
\end{verbatim}
Using these decoupled rules, we specify all the capabilities of each component. The memory is capable of updating with any values of \verb"A", \verb"D", and \verb"D''". The processor is capable of taking its step and, if run isolated, would gladly do it pretending some writing has indeed been performed. The two parts need to be synchronised. Our synchronisation mechanism restricts the wild capabilities of the components. It represents the wiring or gearing in real world systems. Such refined realism is not always needed. When it is used, each component system can be simpler, and can be given independent meaning, independent specification, and, hopefully, independent analysis.

\subsection{Egalitarianism}
\label{egalitarianism}

Using actions for synchronisation, rather than states, is in many cases the natural choice, as illustrated in the two examples above. Thus, it is unfortunate that actions, or transitions, are often treated in a discriminatory fashion with respect to states. In the basic formulations of labelled transition structures, for example, it is usually the case that we can define propositions on states, but actions are only given atomic and non-unique identifiers. In rewriting logic, states are represented by terms of any complexity, but rules are only given atomic labels. This provides little flexibility for dealing with actions.

Synchronising states is also useful, so we don't want to adopt an action-only formalism. We need to be \emph{egalitarian}. In~\cite{Martin2016} we already argued for the convenience of treating states and transitions as equals, and we proposed what we already called \emph{egalitarian} systems. Those systems were different, though related, to the ones by that adjective in the present paper. The same point had been made before, if only partially, for instance in~\cite{DeNicola1995,Kindler1998,Meseguer2008}. Be our task either the specification of systems or of their temporal properties, it can be made simpler and more natural with an egalitarian view. In~\cite{Martin2016b} we showed that also the synchronous composition of systems benefits from being egalitarian.

From the point of view of transition structures, this means that we are going to use propositions (or, rather, properties) both for states and for transitions. From a rewriting logic point of view, this means that we represent transitions (as well as states) by terms. Proof terms, as described in~\cite{Meseguer1992}, can be used in rewriting logic to represent transitions. But, to be egalitarian and to achieve our goals in this paper, proof terms are not appropriate, and we need to somewhat redefine the very concept of transition, as we do next.

\subsection{What is a transition?}
\label{what-is-a-transition}

In our toy computer architecture example, consider the way to deal with a reading instruction \verb"(r, A)", that is, a request to obtain the value stored in memory address \verb"A". In Maude-like syntax, the processor part could be written like this:
\begin{verbatim}
--- In module PROCESSOR
   rl [doingR] : ProgCounter: N
                 Instr: (r, A)
                 Data: D
              => ProgCounter: N + 1
                 Instr: void
                 Data: D' .
\end{verbatim}
That is, at the start, the processor has got an instruction to read the contents of memory address \verb"A" and, after \emph{performing} the reading, it has stored in its register the new value \verb"D'". Different instances of the rule represent different transitions. Again, for any of these transitions to be meaningful, they have to interact, to be synchronised, with some actions at the memory side, but we do not care about these right now.

The point to note here is that each transition, each instance of the rule \verb!doingR!, represents a process of reading, that we can picture as taking place over a certain time span, and the particular value of \verb"D'" is only available \emph{after} the execution of such a process, not \emph{while} executing it. Thus, \verb"D'" cannot be an attribute of a transition represented by the rule \verb!doingR!. It is only an attribute of the destination state. Therefore, the same transition based on \verb!doingR! can take the system to any of a set of destination states, each with a different value for \verb"D'". For similar reasons, \verb"D" needs not be known to the transition, but only to the origin state.

Such is the rationale behind our concept of \emph{transition}. This is our proposal: whenever from any state in a set of origin states $\{o_1,o_2,\dots\}$ a system can reach in one step any state in a set of destination states $\{d_1,d_2,\dots\}$, it is fair to consider such steps the same transition, irrespective of the actual origin and destination states used in each actual run of the system.
\begin{center}
\begin{tikzpicture}[auto]
\node [state] (o1) at (0, 3) {$o_1$};
\node [state] (o2) at (1.5, 3) {$o_2$};
\node         (oi) at (3.4, 3) {\strut\dots};
\node [trans] (t) at (2, 1.5) {};
\node [state] (d1) at (0, 0) {$d_1$};
\node [state] (d2) at (1.5, 0) {$d_2$};
\node         (di) at (3.4, 0) {\strut\dots};
\draw [->, out=-45, in=110] (o1) edge (t);
\draw [->, out=-60, in=100] (o2) edge (t);
\draw [->, out=-135, in=70] (oi) edge (t);
\draw [->, out=-110, in=45] (t) edge (d1);
\draw [->, out=-100, in=60] (t) edge (d2);
\draw [->, out=-70, in=135] (t) edge (di);
\end{tikzpicture}
\end{center}

For another example, consider the rule for updating the memory, that we wrote above and repeat now:
\begin{verbatim}
--- In module MEMORY
   rl [updating] : (A, D'') RestOfMemory
                => (A, D) RestOfMemory .
\end{verbatim}
This represents an updating of the data stored in memory address \verb"A". The data \verb"D''" is better seen as a property belonging to the origin state, irrelevant once the updating process starts. However, the new data \verb"D" is better seen as a property belonging to the transition (and to the destination state as well), because the updating process needs to work with the new value since the moment it starts. Indeed, at the processor side, the rule \verb!doingW!, the one that has to synchronise with \verb"updating", has the new value \verb"D" available already in its origin state, as part of the instruction to be executed: \verb"(w, A, D)".

Transitions are, thus, freed from the usual ``single origin and single destination'' convention. A proof term, as defined in~\cite{Meseguer1992} for rewriting logic, univocally identifies a rewriting step, and contains all the information to recover the origin and destination state terms. For example, \verb!updating(A, D'', RestOfMemory, D)! is the form of a proof term for the rule above. This is too restrictive for us. We want to remove from a transition term all the information that belongs rather to their origin or destination states. The transition term we need in this case is \verb!updating(A, RestOfMemory, D)!. From this term it is not possible to recover the particular origin state, but it is possible to recover the set of possible origin (and destination) states. In the case of rule \verb!doingR! above, the transition term needed is \verb!doingR(A, RestOfMemory)!. It may seem odd that some of these variables are part of the transition terms, specially \verb!RestOfMemory!, but it is the only way that the sets of possible origin and destination states can be recovered. These variables may be better interpreted as the context in which the transition is taking place, instead of being a parameter of the transition.

This feature is shared by states---indeed, it is a trivial thing: for each state, any transition that takes the system to it can be followed by any transition that takes the system from it. Thus, our concept of transition is egalitarian, which is nice. But the real reason for choosing this concept of transition is that we need it for synchronisation to work properly. (Interestingly, the proof terms proposed in~\cite{Boudol1988} for CCS forget the unused part of an alternative, getting thus a little closer to our proposal.)

We said that it is \emph{fair} to consider that some steps are instances of the same transition; we didn't say it is \emph{mandatory}. In the motivational example of Section~\ref{motivation-mutex}, \verb!crit(1)! and \verb!crit(2)! are different transitions, even though they share their unique origin and destination state, \verb"rem". We need them to be different transitions so that each can be synchronised to the action of a different train.

In agreement with all this, in the egalitarian transition systems we define below, transitions are represented by boxes with in and out arrows, as in the diagram above. (The graphical aspect is similar to a Petri net, but the workings are different: in a Petri net, a transition fires only when all its in-places are marked, and then all its out-places get marked; in egalitarian transition structures, only one in-state needs to be \emph{marked} for the transition to \emph{fire}, and only one out-state gets \emph{marked}.)

Translating all of this to rewriting logic leads us to some tweaking in the definition of the logic. It is fully described below. The idea is that, in egalitarian rewriting logic, we need to label rules, not with atomic labels, but with terms showing which parameters we have chosen to belong to the transition.

\subsection{Synchronising on properties}
\label{ppties}

Consider once again the rules \verb!doingW! and \verb!updating! from Section~\ref{motivation-comp-arch}, giving rise to transitions with terms of the form \[\texttt{doingW(N, A, D)} \qquad\textrm{and}\qquad \texttt{updating(A, RestOfMemory, D),}\]
respectively. Let us choose some concrete values for the initial states, as shown in this picture:
\begin{center}
\begin{tikzpicture}[auto, every text node part/.style={align=left}]
\node at (0, 4.2) {\textsc{Processor}};
\node [state] (proc-orig) at (0, 3) {\verb"ProgCounter: 1"\\
                                     \verb"Instr: (w, 1, 5)"\\
                                     \verb"Data: void"};
\node [trans] (proc-trans) at (0, 1.5) {\texttt{doingW(1, {}1, {}5)}};
\node [state] (proc-dest) at (0, 0) {\verb"ProgCounter: 2"\\
                                     \verb"Instr: void"\\
                                     \verb"Data: 5"};
\draw [->] (proc-orig) edge (proc-trans);
\draw [->] (proc-trans) edge (proc-dest);

\node at (7, 4.2) {\textsc{Memory}};
\node [state] (mem-orig) at (7, 3) {\verb"(0, 0) (1, 0)"};
\node [trans] (mem-trans1) at (4.3, 1.5) {\texttt{updating}\\\texttt{(0, (1, 0), 0)}};
\node [state] (mem-dest1) at (4.3, 0) {\verb"(0, 0) (1, 0)"};
\node         (mem-trans2) at (6.3, 2.2) {\dots};
\node         (mem-dest2) at (6.3, 0) {\dots};
\node [trans] (mem-trans3) at (8.3, 1.5) {\texttt{updating}\\\texttt{({}1, (0, 0), {}5)}};
\node [state] (mem-dest3) at (8.3, 0) {\verb"(0, 0) (1, 5)"};
\node         (mem-trans4) at (9.3, 2.2) {\dots};
\node         (mem-dest4) at (10, 0) {\dots};
\draw [->,out=185,in=75] (mem-orig) edge (mem-trans1);
\draw [->,out=-15,in=90] (mem-orig) edge (mem-trans3);
\draw [->] (mem-trans1) edge (mem-dest1);
\draw [->] (mem-trans3) edge (mem-dest3);
\end{tikzpicture}
\end{center}
In words, the processor's state is telling: that the next instruction to be executed is the one identified with number 1 in the program; that the processor has already received and stored such an instruction, which is a request to write 5 at memory address 1; and that the data register is void (maybe because we have just begun executing the program). From there, the processor can only perform the transition shown, reaching the state shown. About the memory, we assume for simplicity that it only stores two data, at addresses 0 and 1. The initial state stores two zeros. From there, the memory is ready to perform any transition following the rule \verb"updating". Two possible transitions and destination states are shown in the picture.

When the memory and the processor are composed and run synchronised, we want the memory to execute only the transition that matches the values from the processor transition term: the \texttt{{}1} and the \texttt{{}5}. In this case, we need to synchronise on a pair of natural numbers. In other cases, more complex data may be needed for synchronisation.

Thus, we use general functions on states and on transitions. We refer to these functions as \emph{properties}. For synchronisation, we need to compare the values of such properties for equality. We require their values to allow for such a comparison. No other restriction is made on the types of properties. We call \emph{synchronisation criterion} to each equality of properties required for composing two systems.

There is one more tweak: properties may be undefined at some states and transitions. In the example above, the processor needs two properties with names, say, \verb!addressBeingAccessed! (the {}1) and \verb!dataBeingWritten! (the {}5). What can the value of \verb!dataBeingWritten! be while the processor is reading or doing something else different from writing? That property makes no sense at those points. No value can be assigned to it. Thus, we do not require that each state and transition assigns a value to each property. When a property is undefined at a state or transition, it does not impose any conditions to synchronise (regarding that property---there may be others).

Partially defined properties are of practical interest. At states or transitions where the property is undefined, other systems can behave in any way. This can be seen as an implication: ``when I say \emph{yes}, you say \emph{yes}; but when you say \emph{yes}, I can say whatever, or nothing at all'' (because I am allowed to be visiting a stage where the property is undefined). This is illustrated in the example in Section~\ref{simple-example} (with the property \verb!doMoveR!).

Systems with partially defined propositions have been more thoroughly studied, in a different setting, in~\cite{Huth2001,Bruns1999,Bruns2000,Godefroid2005}, under the name of \emph{partial structures}. These are Kripke structures on whose states propositions are defined with possible values true, false, or unknown. The \emph{unknown} value can be used, for example, to perform model checking leaving part of the state space unexplored; or to represent a satisfiability problem as a Kripke structure in which all propositions are unknown at start.

\subsection{Synchronising states with transitions}
\label{sync-states-trans}

The rule \verb!doingW! at the processor side, as it has already been discussed, must be synchronised with whatever actions perform the real updating at the memory side. The processor can not know precisely what is happening in the memory, and must not care about it. In our implementation above there was just one rule for updating the memory but, if the specification of the memory were more refined, there may well be several rules involved. It must be possible to synchronise one transition at one side with more than one at the other side, just based on the values of properties.
\begin{center}
\begin{tikzpicture}[auto, every text node part/.style={align=left}]
\node at (0, 5.8) {\textsc{Processor}};
\node [state] (ao) at (0, 4.2) {before requesting update};
\node [trans] (at) at (0, 2.4) {requesting update};
\node [state] (ad) at (0, 0.6) {after requesting update};
\draw [->] (ao) edge (at);
\draw [->] (at) edge (ad);

\node at (6, 5.8) {\textsc{Memory}};
\node [state] (bo) at (6, 4.8) {before updating};
\node [rectangle,draw,minimum height=3.5cm,minimum width=4cm] (bx) at (6, 2.4) {};
\node [trans] (bt1) at (6, 3.6) {first updating step};
\node [state] (bi) at (6, 2.4) {between steps};
\node [trans] (bt2) at (6, 1.2) {second updating step};
\node [state] (bd) at (6, 0) {after updating};
\draw [->] (bo) edge (bt1);
\draw [->] (bt1) edge (bi);
\draw [->] (bi) edge (bt2);
\draw [->] (bt2) edge (bd);

\draw [dashed,out=0,in=180] (ao) edge (bo);
\draw [dashed,out=0,in=180] (ad) edge (bd);
\draw [dashed,out=0,in=180] (at) edge (bx);
\end{tikzpicture}
\end{center}

In our case, the \emph{before} states on both systems must synchronise; also the \emph{after} states must synchronise; and the unique transition at the processor side must synchronise with two transitions and one state at the memory side. In particular, one process (the memory) is in a state (``between steps'') while the other (the processor) is in a transition (``requesting update''). This kind of heterogeneous synchronisation is unavoidable. Therefore, the norm is: either states or transitions in one system can be visited at the same time as either states or transitions in the other system if they agree on the values of their properties. 

The final consequence is that the boundary between states and transitions in composed systems disappears or is blurred. For, what can we say about a system with two components, one in a state while the other in a transition? The composition is not purely in a state and not purely in a transition. When all the components happen to be simultaneously in states (resp., transitions), it is still correct to say that the composed system is in a pure state (resp., transition). Correct, but devoid of any importance. We become, in this way, utterly egalitarian. When we want to refer to either states or transitions or any composition of them, we call them \emph{stages}.

\subsection{The split}

At the start of this paper, we envisioned egalitarian transition structures as represented by bipartite graphs, with states and transitions interleaved. After the discussion above, these structures are only valid as \emph{atomic} components. Non-atomic structures are given as a set of atomic ones together with the synchronisation criteria.

We have stepped into a problem. Our goal for powerful compositionality has driven us to use structures and systems (and sets of them) that dangerously depart from the standard ones. Our second goal (using existing machinery) is compromised. Maude's model checker, for instance, works by traversing a Kripke structure that is a model of the specification, so it is not usable on non-standard structures. Fortunately, there is a way out; we call it the \emph{split} operation.

Splitting applies both to transition structures and to rewrite systems, either atomic or composed. Splitting a structure or a system produces a standard one with equivalent information and behaviour. Splitting a structure or a system consists in making each state, transition, or whatever mixture of them in the original composed system into a state (of a new sort) in the resulting one. An egalitarian transition structure, through splitting, becomes a standard one. An egalitarian rewrite system, through splitting, becomes a standard rewrite system. The reason for the name \emph{split} is that a labelled rewrite rule like $\ell : t \rightarrow t'$ in an egalitarian rewrite system becomes split as $t \rightarrow \ell$ and $\ell \rightarrow t'$. What was a rule label, $\ell$, becomes an intermediate state term in the split system.

The split of an atomic transition structure consists in turning transitions into new states:
\begin{center}
\begin{tikzpicture}[auto]
\node [state] (t1) at (0, 3) {$t_1$};
\node [trans] (l1) at (0, 1.5) {$\ell_1$};
\node [state] (t2) at (-0.6, 0) {$t_2$};
\node [state] (t3) at (0.6, 0) {$t_3$};
\draw [->] (t1) edge (l1);
\draw [->] (l1) edge (t2);
\draw [->] (l1) edge (t3);

\draw [|->] (1.5, 1.5) -- node {$\splitop$} (2.7, 1.5);

\node [state] (t1') at (4, 3) {$t_1$};
\node [state] (l1') at (4, 1.5) {$\ell_1$};
\node [state] (t2') at (3.4, 0) {$t_2$};
\node [state] (t3') at (4.6, 0) {$t_3$};
\draw [->] (t1') edge (l1');
\draw [->] (l1') edge (t2');
\draw [->] (l1') edge (t3');
\end{tikzpicture}
\end{center}

The following is an example for a composed structure. So that the resulting structure is small, we are assuming that properties have been defined in both component structures that allow states to be synchronised only with states, and transitions only with transitions:
\begin{center}
\begin{tikzpicture}[auto]
\node [state] (t1) at (0, 3) {$t_1$};
\node [trans] (l1) at (0, 1.5) {$\ell_1$};
\node [state] (t2) at (-0.6, 0) {$t_2$};
\node [state] (t3) at (0.6, 0) {$t_3$};
\draw [->] (t1) edge (l1);
\draw [->] (l1) edge (t2);
\draw [->] (l1) edge (t3);

\node at (1.1, 1.5) {\large $\|$};

\node [state] (t1') at (2.6, 3) {$t'_1$};
\node [trans] (l1') at (2, 1.5) {$\ell'_1$};
\node [trans] (l2') at (3.2, 1.5) {$\ell'_2$};
\node [state] (t2') at (2.6, 0) {$t'_2$};
\draw [->] (t1') edge (l1');
\draw [->] (t1') edge (l2');
\draw [->] (l1') edge (t2');
\draw [->] (l2') edge (t2');

\draw [|->] (4, 1.5) -- node {$\splitop$} (5.2, 1.5);

\node [state] (t1t1') at (7, 3) {$\langle t_1, t'_1\rangle$};
\node [state] (l1l1') at (6.2, 1.5) {$\langle \ell_1, \ell'_1\rangle$};
\node [state] (l1l2') at (7.8, 1.5) {$\langle \ell_1, \ell'_2\rangle$};
\node [state] (t2t2') at (6.2, 0) {$\langle t_2, t'_2\rangle$};
\node [state] (t3t2') at (7.8, 0) {$\langle t_3, t'_2\rangle$};
\draw [->] (t1t1') edge (l1l1');
\draw [->] (t1t1') edge (l1l2');
\draw [->] (l1l1') edge (t2t2');
\draw [->] (l1l1') edge (t3t2');
\draw [->] (l1l2') edge (t2t2');
\draw [->] (l1l2') edge (t3t2');
\end{tikzpicture}
\end{center}
States like $\langle t_1, \ell'_1\rangle$ are not valid in this case, according to the properties we are assuming for synchronisation. There are three more valid states in the result: $\langle t_1, t'_2\rangle$, $\langle t_2, t'_1\rangle$, and $\langle t_3, t'_1\rangle$; these are not reachable from the ones shown.

Ideally, we would have at our disposal execution engines and verification tools that work in a distributed, compositional way on egalitarian systems. This is not the case at present. The intended use of the split operation is to allow us to translate problems posed on egalitarian structures and systems to standard split ones. For instance, temporal properties would be translated and Maude's model checker can be used on the split system to draw conclusions about the original egalitarian one. The detailed way to perform these formula translations will be the subject of a future paper.

\subsection{On concurrency and topmostness}
\label{concurrency}

True concurrency is a given in rewriting logic. When two rewrite rules refer to disjoint subterms of the state term, they can be executed at the same time. A rewrite system is said to be \emph{topmost} if rewrites can only happen on the whole state term. Topmostness prevents concurrency: if each rule uses the whole state term, there is no room left for another rule to act concurrently on a different subterm. We revisit topmostness in Section~\ref{formal-definitions}.

In our setting, concurrency has one additional level. In a composed system, rules from different components can be executed at the same time, irrespective of whether each atomic component happen to be topmost or not. Carefully chosen values for properties can mandate simultaneity of actions in different components, or can prevent it. When no restrictions are made, both concurrency and interleaving are possible. Thus, true concurrency is natural if we are talking about actions taking place in different components.

Indeed, we prefer that our atomic systems are topmost. Suppose, for example, that the states of our system are given by pairs \verb"(a, b)", and that there are rules to rewrite the complete state term  and also for a part of it:
\begin{verbatim}
(a, b) => (a', b') .
a => a'' .
\end{verbatim}
In many cases, this can be interpreted as a hint that \verb"a" represents a meaningful component, able to evolve by itself. In consequence, we propose to decompose the system into two topmost ones: one for the \verb"a" part, with rules
\begin{verbatim}
a => a' .
a => a'' .
\end{verbatim}
and another for the \verb"b" part, with rule
\begin{verbatim}
b => b' .
\end{verbatim}
We need to add the synchronisation criteria to make rules \verb"a => a'" and \verb"b => b'" be run only simultaneously, so as to produce the same result as the initial composed rule. We have found this method useful when specifying our examples, and we expect it to allow to transform many interesting systems into synchronous compositions of topmost ones.

\subsection{Explicit synchronisation criteria}
\label{syncing-criteria}

In automata theory and in process algebras it is common to define the synchronous composition so that the components execute the same action at the same time. Here, ``the same action'' means ``an action with equal name.'' Instead of relying on implicit understandings, we prefer to be explicit about which properties synchronise with which others in our synchronous compositions. The way we have chosen is the following. A synchronous composition is denoted as $\cals_1 \|_Y \cals_2$. It has a parameter in addition to the systems being composed: a set $Y$ of \emph{synchronisation criteria} specifying properties to synchronise on.

Each criterion is given as a pair $(p_1,p_2)$, with $p_1$ a property from $\cals_1$ and $p_2$ from $\cals_2$. Each criterion $(p_1,p_2)$ is a requirement that, at every moment, the value of $p_1$ at the current state or transition in system $\cals_1$ is the same as the value of $p_2$ at the current state or transition in system $\cals_2$.

In the example about mutual exclusion from Section~\ref{motivation-mutex}, we would define a property \verb"grants(1)" in \verb!MUTEX! as true in the transition \verb!crit(1)! and false otherwise. At \verb!TRAIN1! we would define a property \verb!isCrossing! to be true at its transition \verb!crossing! and false at all other transitions and states. Then, we obtain the desired result by requiring the systems to synchronise on those properties. It is often the case that the same property seen from different sides deserves different, descriptive names. Also, it helps reusability that the \verb!MUTEX! system does not use train-related names.

A composition with no synchronisation is acceptable, that is, with $Y=\emptyset$. It means that each component can evolve with no consideration to what the other one is doing. Also, the same property can be used several times: $Y=\{(p_1,p_2),(p_1,p'_2)\}$. In Section~\ref{formal-definitions} we define the synchronous composition as a binary operation, but then justify expressions in which more than two systems can be composed at once.

\subsection{Properties of the composed system}
\label{composed-ppty}

A way is needed to assign properties to composed systems, so that they can be used as components in turn. Although we have mentioned encapsulation as a desirable ingredient of compositionality, we do not show in this paper the means to facilitate it. (We see in parameterised programming an ideal tool to deal with it in rewriting logic, and we have detailed this proposal in~\cite{MartinVM2018}.) For the purpose of the present paper, any property $p_1$ defined in a system $\cals_1$ is automatically a property of any composed system of which $\cals_1$ is a component, $\cals_1\|_Y\cals_2$, even if $p_1$ is already used in $Y$.

We insist that different systems are different namespaces, so that two properties that happen to get the same name in different systems are different properties. We use the name of the system as a subscript if disambiguation is needed: $p_\cals$ for property $p$ defined in system $\cals$.

\subsection{Summing up}

We have up to now described our motivation, goals, choices, their rationales and some of their consequences. The result is that a system specification is formed by the specifications of a series of \emph{atomic}, maybe topmost, components. In them, we are equally interested in states and transitions. On each system, properties are defined. They work as interfaces, or ports, or handles, or wires, or communication buses, or just attributes, different metaphors being appropriate in different cases. Properties may be partial functions, they need not be defined in all states and transitions. Then, synchronisation criteria are specified as pairs of properties whose values, if defined, must be equal for all component systems at all times. The resulting composed systems can, in their turn, be composed. Properties are inherited by the composed systems from each of its components. This is all formalized in the next section, both for transition structures and for rewrite systems.

\section{Formal definitions and results}
\label{formal-definitions}

We define and use below a number of structures and systems to which we give special names and abbreviations. This is a list of them:
\begin{description}
\item{\textsf{atEgRwSys}}: atomic egalitarian rewrite systems
\item{\textsf{EgRwSys}}: egalitarian rewrite systems
\item{\textsf{RwSys}}: plain rewrite systems
\item{\textsf{atEgTrStr}}: atomic egalitarian transition structures
\item{\textsf{EgTrStr}}: egalitarian transition structures
\item{\textsf{TrStr}}: plain transition structures
\end{description}

\noindent This diagram shows the whole set of structures and systems with their related maps:
\begin{center}
\begin{tikzpicture}[auto, every text node part/.style={align=left}]
\node [rectangle] (atEgTrStr) at (0.5, 0) {\textsf{\upshape atEgTrStr}};
\node [rectangle] (EgTrStr) at (3, 0) {\textsf{EgTrStr}};
\node [rectangle] (EgTrStr') at (6.5, 0) {\textsf{EgTrStr}};
\node [rectangle] (TrStr) at (4, 2) {\textsf{TrStr}};
\node [rectangle] (TrStr') at (7.5, 2) {\textsf{TrStr}};

\node [rectangle] (atEgRwSys) at (0.5, 3.5) {\textsf{\upshape atEgRwSys}};
\node [rectangle] (EgRwSys) at (3, 3.5) {\textsf{\upshape EgRwSys}};
\node [rectangle] (EgRwSys') at (6.5, 3.5) {\textsf{\upshape EgRwSys}};
\node [rectangle] (RwSys) at (4, 5.5) {\textsf{\upshape RwSys}};
\node [rectangle] (RwSys') at (7.5, 5.5) {\textsf{\upshape RwSys}};

\draw (EgRwSys) edge [split] (RwSys);
\draw (EgRwSys') edge [split] (RwSys');

\draw (EgTrStr) edge [split] (TrStr);
\draw (EgTrStr') edge [split] (TrStr');

\draw (atEgRwSys) edge [sem] (atEgTrStr);
\draw (EgRwSys) edge [sem] (EgTrStr);
\draw (RwSys) edge [sem] (TrStr);
\draw (RwSys') edge [sem] (TrStr');

\draw (TrStr) edge [sc] (TrStr');
\draw (EgTrStr) edge [sc] (EgTrStr');
\draw (RwSys) edge [sc] (RwSys');
\draw (EgRwSys) edge [sc] (EgRwSys');

\draw (EgRwSys') edge [sem] (EgTrStr');

\draw (-1, 6) edge [sc] node [below] {$\|$} (0, 6);
\draw (-1, 6) edge [sem] node [left] {$\semop$} (-1, 5);
\draw (-1, 6) edge [split] node {split} (-0.3, 6.7);

\draw [right hook-latex, line width=0.3mm] (atEgRwSys) edge (EgRwSys);
\draw [right hook-latex, line width=0.3mm] (atEgTrStr) edge (EgTrStr);
\end{tikzpicture}
\end{center}
Slanted dashed arrows represent the several concepts of split. Double horizontal arrows represent synchronous composition of systems or structures. Downward snake arrows represent semantic maps. It is argued below that each face of that polyhedron contains a commutative diagram (namely, see Definitions~\ref{split-ets} and~\ref{sem-ers}, Proposition~\ref{prop-sem-rs}, Theorem~\ref{split-sem}, and Section~\ref{split-ers}).

\subsection{Transition structures}
\label{tx-str}

We are interested in three types of transition structures. In atomic egalitarian transition structures, transitions have independent entity; they are, otherwise, similar to Kripke structures and other known formalisms. Then, non-atomic egalitarian transition structures are sets of atomic ones evolving in a synchronised way according to given criteria. Finally, plain structures are useful as a non-egalitarian but quite standard equivalent to egalitarian ones.

\subsubsection{Atomic egalitarian transition structures}
\label{aets}

We denote by \textsf{\upshape atEgTrStr} the class of \emph{atomic egalitarian transition structures}. An element of \textsf{\upshape atEgTrStr} contains two kinds of nodes, called states and transitions, that can only occur interleaved.

Formally, an atomic egalitarian transition structure is a tuple $\calt=(Q,T,{\rightarrow},P,g^0)$, where:
\begin{itemize}
\item $Q$ is the set of states (we reserve the letter $S$ for later use);
\item $T$ is the set of transitions;
\item ${\rightarrow} \subseteq (Q\times T) \union (T\times Q)$ is the bipartite adjacency relation;
\item $P$ is the set of properties---partial functions defined on states and transitions, with possibly different codomains: $p : Q\union T \pfun C_p$ (using $C_p$ to denote the codomain of property $p$);
\item $g^0 \in Q\union T$ is the initial state or transition.
\end{itemize}
The adjacency relation allows for several arrows in and out of a transition, as well as a state. Also, the egalitarian goal mandates that not only an initial state is possible, but also an initial transition. Think of it as if we start studying the system when it is already doing something. We introduce the word \emph{stage} to refer to either a state or a transition. We use variables typically called $g$, with or without subscripts or superscripts, to range over $Q\union T$, that is, over stages. Thus, $g^0$ is the initial stage.

These structures can be seen as generalisations of other, well-known ones. Standard transition structures are obtained as the particular cases in which each transition has a unique origin and a unique destination, that is, in which the restricted relations ${\rightarrow}|_{T\times Q}$ and ${\rightarrow}^{-1}|_{T\times Q}$ are functions. Labelled transition structures are then obtained by adding a single property $L:T\rightarrow\Lambda$ (undefined on states), for some alphabet of actions $\Lambda$. Similarly, Kripke structures are obtained by adding instead a single property $L' : Q \rightarrow 2^\text{AP}$ (undefined on transitions), for some set of atomic propositions $\text{AP}$. As discussed in Section~\ref{what-is-a-transition}, Petri nets are different to atomic transition structures. Even though they are both bipartite graphs, their runs are different. There is no concept of marking in our structures, no distributed state. Our means to get simultaneous firing of transitions is to compose several atomic structures, as we show in the next section.

\subsubsection{Egalitarian transition structures}
\label{ets}

The class of \emph{egalitarian transition structures}, \textsf{EgTrStr}, is built up from atomic transition structures by performing pairwise synchronous compositions. Formally, we define egalitarian transition structures and, simultaneously, their sets of properties, by:
\begin{itemize}
\item if $\calt=(Q,T,\rightarrow,P,g^0)\in\textsf{\upshape atEgTrStr}$, then $\calt\in\textsf{\upshape EgTrStr}$, with $P$ being its set of properties;
\item given $\calt_1,\calt_2\in\textsf{\upshape EgTrStr}$ (atomic or otherwise), with respective sets of properties $P_1$ and $P_2$, and given a relation $Y\subseteq P_1\times P_2$ (the synchronisation criteria), the expression $\calt_1\|_Y\calt_2$ denotes a new element of \textsf{EgTrStr}, called the \emph{synchronous composition of $\calt_1$ and $\calt_2$ with respect to $Y$}, whose set of properties is defined to be $P_1\uplus P_2$.
\end{itemize}
For each synchronisation criterion $(p_1,p_2)\in Y$, we require that the codomain sets of $p_1$ and $p_2$ are the same, $C_{p_1}=C_{p_2}$, so that their values can be compared for equality. The class \textsf{EgTrStr} is the smallest one obtained by applying these two rules.

All properties from $\calt_1$ and from $\calt_2$ are automatically properties of their synchronous composition. Ultimately, all properties are defined in atomic structures, and inherited by the composed structures in which they take part. We assume that we have access to properties defined in any atomic component of an egalitarian transition structure, at any level of nesting.

A note on namespaces: we see each egalitarian transition structure as an independent system. Thus, it may happen, as if by chance, that a property in $\calt_1$ has the same name as another in $\calt_2$, but they would be different properties anyway. That is why we have used the disjoint union symbol above. When disambiguation is required, we use the name of the structure as a subscript: $p_\calt$, $p_{\calt_1}$. Also the sets of states and transitions are assumed disjoint between different structures.

The behaviour of each atomic component is restricted by the behaviour of the other components according to the synchronization criteria. Intuitively, the atomic component $\cala$ is allowed to visit its stage $g$ at a given moment in time only if, simultaneously, each \emph{neighbour} $\cala'$ of $\cala$ is visiting a stage $g'$ that is \emph{compatible} with $g$. By a \emph{neighbour} of $\cala$ we mean any $\cala'$ for which there is a criterion in $Y$ involving a property from $\cala$ and another from $\cala'$; stages $g$ and $g'$ are \emph{compatible} if all the criteria are satisfied when the properties are evaluated at them. By the way, that stages $g$ and $g'$ are compatible is a necessary but not sufficient condition for them to be visited at the same time, because there can be other allowed combinations: the stage $h$, in the same atomic component as $g$, may also be compatible with $g'$. Each execution of the system must non-deterministically choose one of the allowed ways to go. It is always possible to define properties in a tighter way to achieve the desired behaviour.

We could define concepts of global stage and global transition for the composed system, but we prefer the \emph{distributed} view outlined in the previous paragraph. We consider those structures as models of distributed systems. In a distributed system, with components that may even reside in different locations, the concepts of global state and global transition are unimportant, and we avoid them in this work as much as possible (although the $\splitop$ operation defined in Section~\ref{tr-split} provides them, indirectly but explicitly).

Thus, strictly speaking, a non-atomic egalitarian transition structure is \emph{not} a transition structure, in the traditional sense of having states and an adjacency relation on them. The behaviour of $\calt_1 \|_Y \calt_2$ results from the interaction, as given by $Y$, of the behaviours of $\calt_1$ and $\calt_2$ and, ultimately, of the behaviours of their atomic subcomponents.

We do not need to make formal that intuition of simultaneous stages. We do need, however, a more restricted notion of equivalence: the one given by the different ways of composing the same components. For example, $(\calt_1\|_{Y}\calt_2)\|_{Y'}\calt_3$ and $\calt_1\|_{Y''}(\calt_2\|_{Y'''}\calt_3)$ are equivalent if $Y\union Y'=Y''\union Y'''$. Let us make this formal.

\begin{defn}
\label{atomic}
The \emph{set of atomic components} of an egalitarian transition structure is:
\begin{itemize}
\item $\mathrm{atomic}(\calt) = \{\calt\}$ if  $\calt\in\textsf{\upshape atEgTrStr}$,
\item $\mathrm{atomic}(\calt_1 \|_Y \calt_2) = \mathrm{atomic}(\calt_1)\union\mathrm{atomic}(\calt_2)$.
\end{itemize}
\end{defn}

\begin{defn}
\label{criteria}
The \emph{total set of criteria} of an egalitarian transition structure is:
\begin{itemize}
\item $\mathrm{criteria}(\calt) = \emptyset$ if $\calt\in\textsf{\upshape atEgTrStr}$,
\item $\mathrm{criteria}(\calt_1 \|_Y \calt_2) = \widetilde{Y} \union \mathrm{criteria}(\calt_1) \union \mathrm{criteria}(\calt_2)$,
\end{itemize}
where $\widetilde{Y}=\{\{x,y\} \mid \langle x,y\rangle\in Y\}$.
\end{defn}

We need to make the tuples from $Y$ into the two-element sets in $\widetilde{Y}$ so that $\widetilde{Y}=\widetilde{Y^{-1}}$ and then $\mathrm{criteria}(\calt_1\|_Y\calt_2) = \mathrm{criteria}(\calt_2\|_{Y^{-1}}\calt_1)$.

\begin{defn}
Two egalitarian transition structures $\calt_1$ and $\calt_2$ are \emph{equivalent}, denoted as $\calt_1\equiv\calt_2$, iff $\mathrm{atomic}(\calt_1)=\mathrm{atomic}(\calt_2)$ and $\mathrm{criteria}(\calt_1)=\mathrm{criteria}(\calt_2)$.
\end{defn}

In the particular case that $\calt_1$ and $\calt_2$ are atomic, we have $\calt_1\equiv\calt_2 \Leftrightarrow \calt_1=\calt_2$.

This equivalence relation can be seen as expressing a concept of \emph{same behaviour}. Equivalent structures are indistinguishable, and we consider them as representing the same system. This allows us to group the atomic components in the most suitable way for a modular design. Consequently, we allow expressions like ${\|_Y}_{i=1}^n\calt_i$, or $\|_{Y,I}\calt_i$ for some index set $I$, or just $\|_Y\calt_i$ if indexes are implied. They stand for any representative of a $\equiv$-equivalence class.

\subsubsection{Plain transition structures}
\label{ts}

What we call \emph{plain transition structures} are, in essence, traditional transition structures on whose states we define properties. These are non-egalitarian---transitions are just given by the adjacency relation between states, and do not play any special role. The reason we introduce these is that, as we show in the next section, egalitarian transition structures can be translated into plain ones, with some benefits.

Formally, a plain transition structure is a tuple $\calt=(Q,{\rightarrow},P,q^0)$, where:
\begin{itemize}
\item $Q$ is the set of states;
\item ${\rightarrow}\subseteq Q\times Q$ is the adjacency relation;
\item $P$ is the set of properties---partial functions defined on states, with possibly different codomains: $p : Q \pfun C_p$;
\item $q^0\in Q$ is the initial state.
\end{itemize}
The class of plain transition structures is denoted by \textsf{TrStr}.

The synchronous composition of two plain transition structures $\calt_i=(Q_i,{\rightarrow}_i,P_i,q_i^0)\in\textsf{\upshape TrStr}$, for $i=1,2$, with respect to a relation $Y\subseteq P_1\times P_2$ (the synchronisation criteria), is denoted by $\calt_1\|_Y\calt_2$ and is defined to be the plain transition structure $\calt=(Q,{\rightarrow},P,q^0)$, where:
\begin{itemize}
\item $Q$ is the set of pairs $\langle q_1,q_2\rangle\in Q_1\times Q_2$ such that, for each $(p_1,p_2)\in Y$, either $p_1(q_1)$ is undefined, or $p_2(q_2)$ is undefined, or $p_1(q_1)=p_2(q_2)$;
\item for $\langle q_1,q_2\rangle, \langle q'_1,q'_2\rangle\in Q$, we have $\langle q_1,q_2\rangle \rightarrow \langle q'_1,q'_2\rangle$ iff either $q_1\rightarrow_1 q'_1$ and $q_2\rightarrow_2 q'_2$, or $q_1\rightarrow_1 q'_1$ and $q_2=q'_2$, or $q_1=q'_1$ and $q_2\rightarrow_2 q'_2$;
\item $P = P_1\uplus P_2$ and, if $p$ is a property defined in $\calt_i$ ($i=1$ or $2$), then it is defined in $\calt$ by $p(\langle q_1,q_2\rangle)=p(q_i)$ or, with qualification to avoid ambiguities, $p_\calt(\langle q_1,q_2\rangle) = p_{\calt_i}(q_i)$;
\item $q^0=\langle q_1^0,q_2^0\rangle$, assumed to be in $Q$ (that is, to satisfy the criteria in $Y$).
\end{itemize}

\begin{defn}
\label{ts-equiv}
Given two plain transition structures $\calt_i=(Q_i,{\rightarrow}_i,P_i,q_i^0)\in\textsf{\upshape TrStr}$, for $i=1,2$, we say that they are \emph{equivalent}, and denote it by $\calt_1\equiv\calt_2$, if there exist two bijections, $f:Q_1\rightarrow Q_2$ and $F:P_1\rightarrow P_2$, such that, for all $q_1,q'_1\in Q_1$ and all $p_1\in P_1$:
\begin{itemize}
\item the adjacency relation is preserved: $q_1\rightarrow_1 q'_1 \Leftrightarrow f(q_1)\rightarrow_2 f(q'_1)$;
\item codomains for properties are preserved: $C_{p_1} = C_{F(p_1)}$;
\item values of properties are preserved: $p_1(q_1)=F(p_1)(f(q_1))$;
\item initial values are preserved: $f(q_1^0)=q_2^0$.
\end{itemize} 
\end{defn}

For example, $\calt_1\|_Y\calt_2 \equiv \calt_2\|_{Y^{-1}}\calt_1$. More in general, different ways to group, reorder, and compose $n$ given structures, distributing and swapping the synchronization criteria accordingly, produce equivalent structures. Thus, as in the previous section, we allow the expressions ${\|_Y}_{i=1}^n\calt_i$, or $\|_{Y,I}\calt_i$, or $\|_Y\calt_i$. Also, we denote by $\langle q_1,\dots,q_n\rangle$ a typical element in $\|_Y\calt_i$.

The operators ``$\|$'' defined for \textsf{TrStr} and for \textsf{EgTrStr} are related (see next section), but are of a different nature: $\calt_1\|_Y\calt_2$ is an irreducible expression in \textsf{EgTrStr}, but can be reduced in \textsf{TrStr} to a single transition structure. This is, indeed, one of the advantages of plain structures: they are always single graphs, even when they are the result of a composition.

In the next section we show how plain transition structures are useful as a more standard equivalent to egalitarian ones. By \emph{more standard} we mean that they are almost Kripke structures, except that we use properties instead of simply Boolean propositions. Thus, for example, the satisfiability of temporal formulas can be defined for plain transition structures in the same way as for Kripke structures, just replacing Boolean propositions by Boolean-valued properties. Finally, if one does not care about egalitarianism, plain transition structures can be taken as a basis for discussing non-egalitarian compositional specification.

\subsubsection{The split and global states}
\label{tr-split}

The \emph{split} is a translation of transition structures: $\splitop : \textsf{\upshape EgTrStr}\rightarrow\textsf{\upshape TrStr}$. The definition is quite straightforward, but its validity relies on a lemma, proved below, to ensure that the sets of properties of $\calt$ and of $\splitop(\calt)$ are the same.

\begin{defn}
\label{split-ets}
\begin{itemize}
\item Given $\calt=(Q,T,{\rightarrow},P,g^0)\in\textsf{\upshape atEgTrStr}$, its split is $\splitop(\calt)=({Q\union T},{\rightarrow},P,g^0)$, that is, stages are transformed into states.
\item For a non-atomic egalitarian structure $\calt_1\|_Y\calt_2$, its \emph{split} is recursively defined by $\splitop(\calt_1\|_Y\calt_2) = \splitop(\calt_1)\|_Y\splitop(\calt_2)$.
\end{itemize}
\end{defn}

\begin{lemma}
Denoting by $\textrm{ppt}(\calt)$ the set of properties of either an egalitarian or plain transition structure $\calt$, we have $\textrm{ppt}(\calt) = \textrm{ppt}(\splitop(\calt))$. Thus, it is valid the use of the same $Y$ on both sides of the equality ``$\splitop(\calt_1\|_Y\calt_2) = \splitop(\calt_1)\|_Y\splitop(\calt_2)$'' in the definition.
\end{lemma}

\begin{proof}
In the base case, trivially $\textrm{ppt}(\calt) = \textrm{ppt}(\splitop(\calt))=P$.
 
Inductively assuming that $\textrm{ppt}(\calt_i)=\textrm{ppt}(\splitop(\calt_i))$, for $i=1,2$, we have:
\begin{align*}
\mathrm{ppt}(\splitop(\calt_1 \|_Y \calt_2))
   &=\mathrm{ppt}(\splitop(\calt_1) \|_Y \splitop(\calt_2))\\
   &=\mathrm{ppt}(\splitop(\calt_1)) \union \mathrm{ppt}(\splitop(\calt_2))\\
   &=\mathrm{ppt}(\calt_1) \union \mathrm{ppt}(\calt_2)\\
   &=\mathrm{ppt}(\calt_1 \|_Y \calt_2).
\end{align*}
\end{proof}

\begin{prop}
\label{prop-equiv-ets}
For $\calt_1,\calt_2\in\textsf{\upshape EgTrStr}$, we have $\calt_1\equiv\calt_2 \Rightarrow \splitop(\calt_1)\equiv\splitop(\calt_2)$.
\end{prop}

\begin{proof}
As $\calt_1\equiv\calt_2$ implies $\mathrm{atomic}(\calt_1)=\mathrm{atomic}(\calt_1)$, either $\calt_1$ and $\calt_2$ are both atomic or none is. In the first case, we must have $\calt_1=\calt_2$, and the result is immediate. In the second case, the result follows easily from the definitions. For example, the states of $\splitop(\calt_1)$ and those of $\splitop(\calt_2)$ are in bijection because both are given by different ways of making pairs with one stage from each component in $\mathrm{atomic}(\calt_1)=\mathrm{atomic}(\calt_2)$, restricted by the same criteria $\mathrm{criteria}(\calt_1)=\mathrm{criteria}(\calt_2)$.
\end{proof}

The $\splitop$ operation can be extended in the usual way to work on the equivalence classes considered in previous sections. Thus, using the notation $[\|_{Y,I} \calt_i]$ to represent $\equiv$-equivalence classes, we have $\splitop([\|_{Y,I}\calt_i]) \subseteq [\|_{Y,I}\splitop(\calt_i)]$, as a restatement of Proposition~\ref{prop-equiv-ets}.

The states in $\splitop(\calt)$ can be seen as global states for the egalitarian structure $\calt$. Also, the adjacency relation in $\splitop(\calt)$ can be seen as representing global transitions of $\calt$ in the following precise sense.

\begin{prop}
Take $\calt=\|_{Y,I}\cala_i\in\textsf{\upshape EgTrStr}$ with each $\cala_i$ atomic, so that $\splitop(\calt) = \splitop(\|_{Y,I}\cala_i) = \|_{Y,I}\splitop(\cala_i)$. Let $g_i$ and $g'_i$ denote stages (not necessarily different) of $\cala_i$ for each $i\in I$. We have that:
\begin{itemize}
\item $\langle g_1,\dots,g_n\rangle$ is a state in $\splitop(\calt)$ iff $g_i$ and $g_j$ are compatible as stages of $\cala_i$ and $\cala_j$ for all $i,j\in I$, that is, the properties defined at them satisfy all the criteria in $Y$;
\item the step $\langle g_1,\dots,g_n\rangle \rightarrow \langle g'_1,\dots,g'_n\rangle$ is valid in $\splitop(\calt)$ iff (both are states in $\splitop(\calt)$, and) for each $i$, either $g_i=g'_i$ or $g_i\rightarrow g'_i$ is a valid step in $\cala_i$, with at least one instance of the latter.
\end{itemize}
\end{prop}

\begin{proof}
Straightforward from the definitions of $\|$ and the split.
\end{proof}

This opens interesting possibilities for verification. Very roughly, the idea is that proving $\cala_i\models\Phi_i$ for all $i$, where the $\Phi_i$ are temporal formulas, may be enough to prove $\splitop(\|_Y\cala_i)\models\Phi$, for some formula $\Phi$ built based on the $\Phi_i$ and $Y$. These issues are left for future work. See also Section~\ref{comp-verif}.

\subsection{Rewrite systems}

We are interested in three variants of rewrite systems. Atomic egalitarian rewrite systems are similar to standard ones, but rule labels are replaced by complex terms, with the same importance as state terms. Then, non-atomic egalitarian rewrite systems are sets of atomic ones interacting through given synchronisation criteria. Finally, plain rewrite systems are useful as a non-egalitarian but quite standard equivalent to egalitarian ones.

The exposition in this section reflects and, in part, repeats the one in Section~\ref{tx-str} concerning transition structures.

\subsubsection{Atomic egalitarian rewrite systems}
\label{aers}

We slightly modify the standard definition of rewrite systems from~\cite{Meseguer1992} to accommodate in it the kind of transitions we are interested in. We call the result \emph{atomic egalitarian rewrite systems} and denote their class by \textsf{\upshape atEgRwSys}. Non-atomic systems are defined below as the composition of atomic ones.

\label{atEgRwSys_def}An atomic egalitarian rewrite system is a tuple $\calr=(S,\le,\Sigma,E,M,R)$, where:
\begin{itemize}
\item $(S,\le,\Sigma)$ is an order-sorted signature. Concretely, $(S,\le)$ is a partial order of sorts and subsorts, and $\Sigma$ is a set of operator declarations of the form $f : s_1 \times \dots \times s_n \rightarrow s$ for $s_1,\dots,s_n,s$ sorts in $S$. The order-sorted signature $(S,\le,\Sigma)$ can be \emph{kind-completed} to $(\hat S,\hat\le,\hat\Sigma)$, by adding to it exactly the following:
\begin{itemize}
\item for each weakly connected component $k$ of $(S,\le)$, include in $\hat S$ an element, that we call $k$ as well, with $s\mathop{\hat\le}k$ for all sorts $s$ in $k$; we denote by $[s]$ the unique such an element with $s\mathop{\hat\le}[s]$, called the \emph{kind} of $s$;
\item for each operator $f : s_1 \times \dots \times s_n \rightarrow s$ in $\Sigma$, include in $\hat\Sigma$ an overloaded operator \emph{at the kind level} $f : [s_1] \times \dots \times [s_n] \rightarrow [s]$.
\end{itemize}
In the rest of the paper, we assume without mention that all the order-sorted signatures we use are already the result of a kind completion. The set $T_\Sigma$ of ground terms is defined in the usual way, as is the set $T_\Sigma(X)$ of terms over a given set of $S$-sorted variables, $X=\{X_s\}_{s\in S}$. Also, $T_{\Sigma,s}$ and $T_\Sigma(X)_s$ denote the sets of terms of sort $s$.

We need atomic rewrite systems to be \emph{topmost}, that is:
\begin{itemize}
\item there are two distinguished sorts in $S$, that we assume called $\texttt{State}$ and $\texttt{Trans}$, to represent states and transitions, respectively; for convenience, we also assume a supersort of both called $\texttt{Stage}$ whose elements we name \emph{stages};
\item the operators and signatures in $\Sigma$ are such that they do not allow a term of sort $\texttt{Stage}$ to include a proper subterm of this same sort.
\end{itemize}

\item $E$ is a set of (possibly conditional) equations. Each equation has the form ``$t=t' \rif C$'' or ``$t=t'$'', where $t,t'\in T_\Sigma(X)$ and $t$ and $t'$ are in the same kind, that is, $[\mathrm{ls}(t)]=[\mathrm{ls}(t')]$, where $\mathrm{ls}(u)$ denotes the \emph{least sort} of the term $u$ in $(S,\le)$. The optional condition $C$ is a conjunction of equational and membership conditions $C=\Big(\bigwedge_i t_i=t'_i\Big) \land \Big(\bigwedge_j t_j:s_j\Big)$, where, for all $i$ and $j$, we have $t_i,t'_i,t_j\in T_\Sigma(X)$, $s_j\in S$, $[\mathrm{ls}(t_i)]=[\mathrm{ls}(t'_i)]$, and $[\mathrm{ls}(t_j)]=[s_j]$. A membership condition ``$t_j:s_j$'' represents the requirement that the term $t_j$ has sort $s_j$.

\item $M$ is a set of (possibly conditional) membership axioms. Each membership axiom has the form ``$t:s \rif C$'' or ``$t:s$'', where $t\in T_\Sigma(X)$, $s\in S$, $[\mathrm{ls}(t)]=[s]$, and the optional condition $C$ is as described in the previous item. The axiom ``$t:s$'' states that the term $t$ has sort $s$.

\item $R$ is a set of (possibly conditional) rewrite rules. Each rewrite rule has the form ``$t \trans{\ell} t' \rif C$'' or ``$t \trans{\ell} t'$'', where $t,t'\in T_\Sigma(X)_\texttt{State}$, $\ell\in T_\Sigma(X)_\texttt{Trans}$, and the condition $C$ is as above. (We do not consider rewrite conditions for the time being.)
\end{itemize}

The transition term $\ell$ replaces the rule labels from standard rewrite systems. We write it in the middle of the rule, both in mathematical notation, as above, and in Maude-like code in the examples from now on. That $\ell$ is a term, and not an atomic identifier, is the main difference between standard and atomic egalitarian rewrite systems. Constructors for transition terms must be declared in $\Sigma$, the same as all other operators. Transition terms can be used in equations and in membership axioms as any other valid terms.

In all cases, unconditional instances can be considered as shortcuts for conditional ones with the condition $C$ being ``true'', which simplifies the presentation of what follows.

The choice of the transition term in a rule and, in particular, of the variables in it is a concern for the specifier, with semantic consequences. For a schematic illustration, consider the rule $t(x)\trans{\dots}t'(y)$, where $x$ and $y$ are the only variables in $t$ and $t'$, resp. We depict below the semantics for five possible transition terms. In each transition term, we show explicitly all variables involved. We assume that the variables $x$, $y$, and $z$ range over $\{0,1\}$.
\begin{center}
\begin{tikzpicture}[auto]
\node at (0.6, 4) {$t(x)\trans{\ell(x)}t'(y)$};
\node [state] (xt0) at (0, 3) {$t(0)$};
\node [state] (xt1) at (1.2, 3) {$t(1)$};
\node [trans] (xl0) at (0, 1.5) {$\ell(0)$};
\node [trans] (xl1) at (1.2, 1.5) {$\ell(1)$};
\node [state] (xt'0) at (0, 0) {$t'(0)$};
\node [state] (xt'1) at (1.2, 0) {$t'(1)$};
\draw [->] (xt0) edge (xl0);
\draw [->] (xt1) edge (xl1);
\draw [->] (xl0) edge (xt'0);
\draw [->] (xl0) edge (xt'1);
\draw [->] (xl1) edge (xt'0);
\draw [->] (xl1) edge (xt'1);

\node at (4.6, 4) {$t(x)\trans{\ell(y)}t'(y)$};
\node [state] (yt0) at (4, 3) {$t(0)$};
\node [state] (yt1) at (5.2, 3) {$t(1)$};
\node [trans] (yl0) at (4, 1.5) {$\ell(0)$};
\node [trans] (yl1) at (5.2, 1.5) {$\ell(1)$};
\node [state] (yt'0) at (4, 0) {$t'(0)$};
\node [state] (yt'1) at (5.2, 0) {$t'(1)$};
\draw [->] (yt0) edge (yl0);
\draw [->] (yt0) edge (yl1);
\draw [->] (yt1) edge (yl0);
\draw [->] (yt1) edge (yl1);
\draw [->] (yl0) edge (yt'0);
\draw [->] (yl1) edge (yt'1);

\node at (8.6, 4) {$t(x)\trans{\ell(z)}t'(y)$};
\node [state] (zt0) at (8, 3) {$t(0)$};
\node [state] (zt1) at (9.2, 3) {$t(1)$};
\node [trans] (zl0) at (8, 1.5) {$\ell(0)$};
\node [trans] (zl1) at (9.2, 1.5) {$\ell(1)$};
\node [state] (zt'0) at (8, 0) {$t'(0)$};
\node [state] (zt'1) at (9.2, 0) {$t'(1)$};
\draw [->] (zt0) edge (zl0);
\draw [->] (zt0) edge (zl1);
\draw [->] (zt1) edge (zl0);
\draw [->] (zt1) edge (zl1);
\draw [->] (zl0) edge (zt'0);
\draw [->] (zl0) edge (zt'1);
\draw [->] (zl1) edge (zt'0);
\draw [->] (zl1) edge (zt'1);
\end{tikzpicture}
\end{center}
\vspace{5mm}
\begin{center}
\begin{tikzpicture}[auto]
\node at (0.6,4) {$t(x)\trans{\ell}t'(y)$};
\node [state] (t0) at (0, 3) {$t(0)$};
\node [state] (t1) at (1.2, 3) {$t(1)$};
\node [trans] (l) at (0.6, 1.5) {$\ell$};
\node [state] (t'0) at (0, 0) {$t'(0)$};
\node [state] (t'1) at (1.2, 0) {$t'(1)$};
\draw [->] (t0) edge (l);
\draw [->] (t1) edge (l);
\draw [->] (l) edge (t'0);
\draw [->] (l) edge (t'1);

\node at (6.15, 4) {$t(x)\trans{\ell(x,y)}t'(y)$};
\node [state] (xyt0) at (5.5, 3) {$t(0)$};
\node [state] (xyt1) at (6.8, 3) {$t(1)$};
\node [trans] (xyl00) at (4.2, 1.5) {$\ell(0,0)$};
\node [trans] (xyl01) at (5.5, 1.5) {$\ell(0,1)$};
\node [trans] (xyl10) at (6.8, 1.5) {$\ell(1,0)$};
\node [trans] (xyl11) at (8.1, 1.5) {$\ell(1,1)$};
\node [state] (xyt'0) at (5.5, 0) {$t'(0)$};
\node [state] (xyt'1) at (6.8, 0) {$t'(1)$};
\draw [->] (xyt0) edge (xyl00);
\draw [->] (xyt0) edge (xyl01);
\draw [->] (xyt1) edge (xyl10);
\draw [->] (xyt1) edge (xyl11);
\draw [->] (xyl00) edge (xyt'0);
\draw [->] (xyl01) edge (xyt'1);
\draw [->] (xyl10) edge (xyt'0);
\draw [->] (xyl11) edge (xyt'1);
\end{tikzpicture}
\end{center}
Focusing only on states, abstracting transitions away, the result is in all cases the same: from any of $t(0)$ or $t(1)$, in one step, any of $t'(0)$ or $t'(1)$ can be reached. However, the importance of transitions in modelling systems is one of the ideas put forward in this work. It is specially discussed in Sections~\ref{egalitarianism} and~\ref{what-is-a-transition}.

The rewrite relation induced by an atomic egalitarian rewrite system is somewhat tricky. At first sight, the natural definition may seem to be this one, involving three terms at once:

\begin{quoting}
$u\stackrel{v}{\rightarrow}_\calr u'$ holds for terms $u,v,u'\in T_\Sigma$ iff there are a rule ``$t \trans{\ell} t' \rif C$'' in $R$, and a substitution $\theta$ such that $\calr\models C\theta$, $u =_E t\theta$, $v =_E \ell\theta$, and $u' =_E t'\theta$ (where the relation $=_E$ is the one induced by the equations in $E$).
\end{quoting}

This definition, however, is not appropriate, for reasons introduced above: first, that would allow information to be passed from the origin to the destination states that is not present in the transition, as in the rule $x\trans{a}x$ for constant $a$ and variable $x$ (see also Section~\ref{what-is-a-transition}); second, these \emph{instantaneous} transitions would make it difficult to synchronise states with transitions, as we sometimes need (see Section~\ref{sync-states-trans}); third, we loose the symmetry between the roles played by states and transitions required by egalitarianism.

Instead, we propose that each instance of a rule gives rise to two \emph{half rewrites}. We need to devote a few paragraphs to discuss this point.

\begin{defn}[half-rewrite relation]
\label{defn:half-rewrite}
Given an atomic egalitarian rewrite system $\calr=(S,\le,\Sigma,E,M,R)$ and two terms $u,v\in T_\Sigma$, we define the \emph{half-rewrite} relation $u\rightarrow^\mathrm{eg}_\calr v$ to hold if either:
\begin{itemize}
\item there are a rule ``$t \trans{\ell} t' \rif C$'' in $R$, and a substitution $\theta_1 : \varsop(\ell,t,t',C) \rightarrow T_\Sigma$ such that $\calr\models C\theta_1$, and $u =_E t\theta_1$, and $v =_E \ell\theta_1$; or
\item there are a rule ``$t \trans{\ell} t' \rif C$'' in $R$, and a substitution $\theta_2 : \varsop(\ell,t,t',C) \rightarrow T_\Sigma$ such that $\calr\models C\theta_2$, and $u =_E \ell\theta_2$, and $v =_E t'\theta_2$.
\end{itemize}
\end{defn}

The first half rewrite can be used to take the system from a state to a transition; the second, from a transition to a state. Both ornaments on the arrow, $^\mathrm{eg}$ and $_\calr$, are usually omitted when they are implied from the context.

Different substitutions can be used for each half rewrite. The substitutions $\theta_1$ and $\theta_2$, when used for two consecutive half rewrites, necessarily agree on $\varsop(\ell)$, but other variables can get different instantiations. This formalizes the fact that, as soon as a transition starts executing, it forgets which particular state it came from. Thus, no data can be passed from the origin state to the destination state, unless it is present in the transition term as well. The scope of variables is limited to half rewrites. A consequence is that the rule ``$t(x) \trans{\ell} t'(x)$'', where $x\not\in\varsop(\ell)$, is valid, but its semantics may not be the one expected: the two occurrences of $x$ are interpreted as different variables, and it may well be that $\theta_1(x)\ne\theta_2(x)$. Indeed, that rule is equivalent to ``$t(x) \trans{\ell} t'(y)$'' (whose semantics was depicted above). This is counter-intuitive, so we prefer that our specifications do not contain such rules. We say that a rule is \emph{readable} when all occurrences of a variable can be interpreted as the same variable.

\begin{defn}[readability]
Consider an atomic egalitarian rewrite system $\calr=(S,\le,\Sigma,E,M,R)$ and a rule in it ``$t \trans{\ell} t' \rif C$''. We say that the rule is \emph{readable} iff for every two consecutive half rewrites $u\rightarrow v$ and $v\rightarrow w$ performed using that rule with substitutions $\theta_1$ and $\theta_2$, respectively, the same two half rewrites can be obtained by using the same substitution for both. An atomic egalitarian rewrite system is \emph{readable} if all its rules are.
\end{defn}

According to this definition, the rule ``$x \trans{a} x$'', for constant $a$ and integer variable $x$, is not readable, because it allows the two half rewrites $3\rightarrow a\rightarrow 4$, that cannot be obtained with a single substitution for $x$. In contrast, the rule ``$x \trans{a} z$'' is readable, as is ``$x \trans{x} x$''. The rule ``$x \trans{y} z \rif x=z$'', for variables $x$, $y$, and $z$, is not readable. The similar ``$x \trans{y} z \rif x=y \land x=z$'' is indeed readable, though we would prefer the equivalent ``$x \trans{y} z \rif x=y \land y=z$'', that does not directly relate $x$ and $z$.

Fortunately, every atomic egalitarian system can be transformed into an equivalent readable one, as shown below. In practice, we prefer readable systems, but our theory treats them all the same, either readable or not. We do not require that a readable rule is always used with the same substitution for two consecutive half rewrites, but just that it \emph{could} be used that way. The point is that some variables in a readable rule are only important for a half rewrite, and can be given any values for the other. For example, the readable rule ``$x \trans{a} z$'' allows rewriting $4\rightarrow a\rightarrow 5$ with $\theta_1(x)=4$ and $\theta_2(y)=5$, choosing whatever values for $\theta_1(y)$ and $\theta_2(x)$, in particular, it is valid to choose 5 and 4, to make $\theta_1=\theta_2$.

\begin{prop}
Every atomic egalitarian rewrite system $\calr=(S,\le,\Sigma,E,M,R)$ can be transformed into a readable one $\calr'=(S,\le,\Sigma,E,M,R')$ (where only the set of rules has changed) that induces the same half-rewrite relation.
\end{prop}

\begin{proof}
We will show that each rule ``$t \trans{\ell} t' \rif C$'' in $R$ can be transformed into a readable one. The first step in the transformation is the following: for each variable, say $x$, in $\varsop(t)\inters\varsop(\ell)$, we choose two fresh variables, say $x_t$ and $x_\ell$, and replace $t$ by $t[x_t/x]$ and $\ell$ by $\ell[x_\ell/x]$, and add a new condition $x_\ell=x_t$ to the set of conditions $C$. We perform the same renaming for variables shared by $\ell$ and $t'$, or by $t$ and $t'$, or by all the three of them. The result is that the sets of variables $\varsop(t)$, $\varsop(\ell)$, and $\varsop(t')$ are pairwise disjoint, and all relations among them are part of the rule conditions. For example, the rule ``$x \trans{a} x$'' is transformed into ``$x_t \trans{a} x_{t'} \rif x_t=x_{t'}$'' (which is not readable yet). The resulting rule is easily seen to induce the same half rewrites as the original one.

Assume, then, that $\varsop(t)$, $\varsop(\ell)$, and $\varsop(t')$ are pairwise disjoint in our rule. Next, we duplicate each condition in $C$. We denote the result by ``$t \trans{\ell} t' \rif C \land C'$'', where $C$ and $C'$ are identical lists of conditions. The intuitive idea is that $C'$ (after some additional work on it) will be only important for the second half rewrite, and $C$ only for the first. For each variable, say $x$, that is in $\varsop(C)\setminus\varsop(t,\ell)$, we choose a fresh variable, say $x'$, and replace $C$ by $C[x'/x]$. In the same way, for each variable, say $y$, that is in $\varsop(C')\setminus\varsop(\ell,t')$, we choose a fresh variable, say $y'$, and replace $C'$ by $C'[y'/y]$. We keep the names $C$ and $C'$ for the results of performing all these renamings. Some conditions may still be present in its exact same form in both resulting $C$ and $C'$, but this is not a problem for readability.

The new rule is semantically equivalent to the original one. Each half rewrite that was possible using the original rule and a substitution $\theta$ is still possible using the new rule and the extended substitution $\theta'$, that is defined on the new variables as they were in the original ones, that is, $\theta'(x')=\theta(x)$, if $x'$ is the new variable chosen to rename $x$. Also, if a half rewrite is possible according to the new rule with a substitution $\theta$, the same half rewrite was possible with the original rule and the corresponding substitution. Our example rule has become ``$x_t \trans{a} x_{t'} \rif x_t=x'_{t'} \land x'_t=x_{t'}$''.

The resulting rule is readable. The only variables that can be in $\varsop(C)\inters\varsop(C')$ are the ones in $\varsop(\ell)$. But, as noted before, $\theta_1$ and $\theta_2$, used for consecutive half rewrites, must coincide in $\varsop(\ell)$. Thus, we can define the common substitution $\theta = \theta_1|_{\varsop(C)} \union \theta_2|_{\varsop(C')}$.
\end{proof}

Another noteworthy consequence of our definition of the half-rewrite relation is that a transition term can be rewritten with a rule that is not the one that produced it. Consider an atomic egalitarian rewrite system with these two rules: $b \trans{a} c$ and $d \trans{a} e$, for constants $a$, $b$, $c$, $d$, and $e$ of appropriate sorts. They specify a single transition identified by the term $a$. The joint semantics is:
\begin{center}
\begin{tikzpicture}[auto]
\node [state] (b) at (0, 2) {$b$};
\node [state] (d) at (1.4, 2) {$d$};
\node [trans] (a) at (0.7, 1) {$a$};
\node [state] (c) at (0, 0) {$c$};
\node [state] (e) at (1.4, 0) {$e$};
\draw [->] (b) edge (a);
\draw [->] (d) edge (a);
\draw [->] (a) edge (c);
\draw [->] (a) edge (e);
\end{tikzpicture}
\end{center}
Once the transition $a$ is executing, it forgets where it came from. Half rewrites $b\rightarrow a\rightarrow e$ must be allowed. This may be acceptable, but one should think twice before using it in the practice of system specification.

Still one more consequence worth mentioning is that the initial term, that we identify as \verb!init!, can be a state term or a transition term. All this is in agreement with our egalitarian view of rewrite systems.

In a rule ``$t \trans{\ell} t' \rif C$'' such that $\varsop(t)\union\varsop(t')\union\varsop(C)\subseteq\varsop(\ell)$, a transition term univocally identifies one origin and one destination state terms. Its rewrite semantics is equivalent to that of the standard-setting rewrite rule ``$t \trans{\ell} t' \rif C$'' considering now $\ell$ as an unstructured atomic label. Transition terms, in this case, are proof terms. Thus, if we care to transform a standard rewrite system into a semantically equivalent atomic egalitarian one, all we need to do is to transform each rule label $\ell$ into an operator that takes as parameters all the variables in $\varsop(t)\union\varsop(t')\union\varsop(C)$ and produces a term of sort \verb"Trans".

A \emph{property} in an atomic system is any unary operator in $\Sigma$ equationally defined, partially or totally, on \verb"Stage" or some of its subsorts. Properties are used below to establish synchronisation criteria for compositions (see Section~\ref{ers}). An important ingredient of our view of compositionality is that properties are defined only at the stages where they are meaningful, and undefined properties pose no requirements for synchronisation. That is why we allow partially defined properties.

Rewriting logic is built on a base equational logic. As is usual since~\cite{Meseguer1992}, we have chosen \emph{membership equational logic}. This is apparent in our definition of atomic rewrite system in page~\pageref{atEgRwSys_def}, that includes membership axioms and membership conditions.

Membership equational logic allows partially defined properties (or  operators in general) by declaring them \emph{at the kind level}. All our properties have to be defined like this: $p : \texttt{Stage} \rightarrow [s]$. Sometimes, $p(t)$ has sort $s$, and we say it is defined at $t$; sometimes, $p(t)$ has just kind $[s]$ but no sort, which we interpret as it being undefined at this particular $t$. It is not necessary to make this explicit in the declaration of properties, however, because in a kind-completed signature each operator is automatically overloaded at the kind level, as already described.

(A different, standard way to represent partial functions is by adding a \emph{bottom} value to its codomain. Whether to use this approach or ours (representing undefined properties by kinded but not sorted expressions) is, to some extent, a matter of personal taste. But we find our approach more appealing both for theory and for implementation, for technical reasons that need not be discussed here.)

\subsubsection{Egalitarian rewrite systems}
\label{ers}

The class of \emph{egalitarian rewrite systems}, \textsf{EgRwSys}, is built up from atomic ones by performing synchronous compositions. Formally, we define egalitarian rewrite systems and, simultaneously, their sets of properties, by:
\begin{itemize}
\item if $\calr\in\textsf{\upshape atEgRwSys}$, then $\calr\in\textsf{\upshape EgRwSys}$, with its properties being its operators partially or totally defined on \verb"Stage" or some subsort;
\item given $\calr_1,\calr_2\in\textsf{\upshape EgRwSys}$, with respective sets of properties $P_1$ and $P_2$, and given a relation $Y\subseteq P_1\times P_2$ (the synchronisation criteria), the expression $\calr_1\|_Y\calr_2$ denotes a new element of \textsf{EgRwSys}, called the \emph{synchronous composition of $\calr_1$ and $\calr_2$ with respect to $Y$}, whose set of properties is defined to be $P_1\uplus P_2$.
\end{itemize}
The class \textsf{EgRwSys} is the smallest one obtained by applying these two rules.

All properties from $\calr_1$ and from $\calr_2$ are automatically properties of their synchronous composition. Ultimately, all properties are defined in atomic systems, and inherited by the composed systems in which they take part. When disambiguation is needed, we refer to a property by $p_{\calr_1}$, or by $p_{\cala}$ if $\cala$ is the atomic system on which the property was initially defined.

Each synchronisation criterion $(p_1,p_2)\in Y$ is interpreted (see below) as a required equational equivalence between the values of $p_1(g_1)$ and $p_2(g_2)$ (if defined) at the respective stages $g_1$ and $g_2$. This implies that some common equational infrastructure in the systems involved is needed. Suppose that, for $i=1,2$, the system $\cala_i$ is the atomic component of $\calr_i$ in which $p_i$ was originally defined, and declared as $p_i : \texttt{Stage}_{\cala_i} \rightarrow s_i$, for $s_i\in S_{\cala_i}$. We require and assume the existence of a membership equational theory $\cale=(S_\cale,\le_\cale,\Sigma_\cale,E_\cale,M_\cale)$, with $s_1,s_2\in S_\cale$ and $[s_1]_\cale=[s_2]_\cale$, which is a common subtheory of the membership equational theories of $\cala_1$ and of $\cala_2$. That is, $\cale$ is component-wise included and protected in both membership equational theories of $\cala_1$ and of $\cala_2$. The word \emph{protected} has here the same meaning as in Maude: the set of ground terms is preserved, that is, for each $s\in S_\cale$, we have $T_{\Sigma_\cale,s}=T_{\Sigma_{\cala_1},s}=T_{\Sigma_{\cala_2},s}$.
If $s=\text{lub}(s_1,s_2)$ is the least common supersort of $s_1$ and $s_2$ in $\cale$, we refer to $\cale$ as \emph{the common theory of $s$}. It is in this theory that the equality of $p_1(g_1)$ and $p_2(g_2)$ is to be checked for synchronisation, namely in the form $p_1(g_1)\!\downarrow_{\cala_1} \;=_\cale\; p_2(g_2)\!\downarrow_{\cala_2}$ (down arrows representing reduction to normal forms, as usual). This is only needed when $p_1(g_1)$ and $p_2(g_2)$ are both defined, that is, they have sorts $s_1$ and $s_2$ respectively, both in $S_\cale$, so that the equational equivalence check makes sense. All this is used a few lines below.

Strictly speaking, a non-atomic rewrite system is \emph{not} a rewrite system, in the traditional sense of inducing a rewrite relation between terms. A non-atomic rewrite system models a distributed system and, thus, its behaviour can only be described based on the interaction between the behaviours of, ultimately, its atomic subcomponents. The behaviour of an atomic component of an egalitarian rewrite system is restricted by the other components according to the synchronization criteria.

The definition of the set of atomic components and the total set of criteria for an egalitarian rewrite system are almost identical to the ones given for transition structures in Section~\ref{ets} (Definitions~\ref{atomic} and~\ref{criteria}), and need not be repeated here. Like there, the relation $\calr_1\equiv\calr_2$ is defined by $\mathrm{atomic}(\calr_1)=\mathrm{atomic}(\calr_2)$ and $\mathrm{criteria}(\calr_1)=\mathrm{criteria}(\calr_2)$, and is interpreted as $\calr_1$ and $\calr_2$ representing the same behaviour, that is, being equivalent models of the same system. We allow expressions like ${\|_Y}_{i=1}^n\calr_i$, or $\|_{Y,I}\calr_i$ for some index set $I$, or just $\|_Y\calr_i$ if indexes are implied, to stand for any representative of a $\equiv$-equivalence class.

\subsubsection{Plain rewrite systems}
\label{rs}

What we call \emph{plain rewrite systems} are, in essence, traditional rewrite systems whose rules have no labels. These are non-egalitarian---transitions do not play any special role. The reason we introduce these is that, as we show in the next section, egalitarian rewrite systems can be translated into plain ones, with some benefits.

In the following definition we avoid repeating what has already been said in Section~\ref{aers} when defining atomic egalitarian systems. A plain rewrite system is a tuple $(S,{\le},\Sigma,E,M,R)$, where:
\begin{itemize}
\item $(S,\le,\Sigma)$ is an order-sorted signature, assumed kind-complete.

\item $E$ is a set of (possibly conditional) equations, each of the form ``$t=t' \rif C$'', with $C=\Big(\bigwedge_i t_i=t'_i\Big) \land \Big(\bigwedge_j t_j:s_j\Big)$.

\item $M$ is a set of (possibly conditional) membership axioms, each of the form ``$t:s \rif C$'', with $C$ as above.

\item $R$ is a set of (possibly conditional) unlabelled rewrite rules. Each rewrite rule has the form ``$t \rightarrow t' \rif C$'', where $t,t'\in T_\Sigma(X)$, and $[\mathrm{ls}(t)]=[\mathrm{ls}(t')]$, and the condition $C$ is as above.
\end{itemize}
The class of plain rewrite systems is denoted by \textsf{RwSys}.

The mandatory absence of labels in rules is the main difference between standard rewrite systems and our plain ones. It has the effect of making transitions identifiable just by their origin and destination states, thus merging all edges between two given states into one.

The (one-step) rewrite relation $u\rightarrow_\calr u'$ in a plain rewrite system is as in the standard case, that is, $u\rightarrow_\calr u'$ holds for terms $u,u'\in T_\Sigma$ iff there are a rule ``$t \rightarrow t' \rif C$'' in $R$, a term $w\in T_\Sigma$, a position $p$ in $w$, and a substitution $\theta$ such that $\calr\models C\theta$, $u =_E w=w[t\theta]_p$, and $u' =_E w[t'\theta]_p$.

We assume that the distinguished sort of the terms we are interested in rewriting is called \verb"State". The initial state is assumed to be called \verb!init!. A \emph{property}, in this context, is any operator in $\Sigma$ partially or totally defined on \verb"State" or some subsort of it. Properties are used to establish synchronisation criteria for compositions.

The operation of synchronous composition for plain rewrite systems is only defined for operands that satisfy certain requirements. Given two plain rewrite systems, $\calr_1$ and $\calr_2$, with respective sets of properties $P_1$ and $P_2$, their synchronous composition with respect to the synchronisation criteria $Y\subseteq P_1\times P_2$, is only defined if (i) $\calr_1$ and $\calr_2$ are topmost and (ii) the properties used in $Y$ are totally defined, each in its respective system. We justify these points.

In rewriting logic a rule $a \rightarrow a'$ can be used to rewrite any term that contains $a$. For instance, the term $f(a,d)$ is rewritten to $f(a',d)$ using that rule. As detailed below, the synchronous composition uses the components' rules to produce composed ones. If a system includes the rule $a \rightarrow a'$ and another system includes $b \rightarrow b'$, then their composition includes the rule $\langle a,b\rangle \rightarrow \langle a',b'\rangle$ (with some conditions), to be used on composed states. But, while $f(a,d)$ can be rewritten by $a \rightarrow a'$, no term of the form $\langle f(a,d),f(b,e)\rangle$ can be rewritten by the composed rule $\langle a,b\rangle \rightarrow \langle a',b'\rangle$, because $\langle a,b\rangle$ is not a subterm of $\langle f(a,d),f(b,e)\rangle$. The way to ensure that rewrites are preserved by composition is asking the systems to be topmost, that is, that rewrites are possible only at the top of state terms.

Syntactic conditions equivalent to topmostness are given, for example, in~\cite{MeseguerT2007}. They are similar to the ones we required for atomic egalitarian rewrite systems at the beginning of Section~\ref{aers}. That same paper shows results on \emph{completing} rules so as to transform a rewrite system into an equivalent topmost one. The idea is to substitute rule $a \rightarrow a'$ by $f(a,x) \rightarrow f(a',x)$ and any other that could be needed for other contexts in which $a$ can occur. Many interesting rewrite systems are amenable to completion in this way. Our view is different. As explained in Section~\ref{concurrency}, instead of completing rules, we prefer to decompose systems until each atomic system is topmost.

The reason to require that all properties used in $Y$ be totally defined is the following. Properties undefined on a particular state pose no conditions for synchronisation. If, for example, the synchronization criteria is $Y=\{(p_1,p_2)\}$, and $p_1(q_1)$ is undefined at a given $q_1\in\texttt{State}_{\calr_1}$, then $q_1$ is compatible with any $q_2\in\texttt{State}_{\calr_2}$, and $\langle q_1, q_2\rangle$ has sort $\texttt{State}_{\calr_1\|_Y\calr_2}$. We run into problems when we need to express this as a valid membership condition. Undefinedness of a term, in membership equational logic, is represented by terms having kind but not sort. While positive sort or kind assertions like $q:[\texttt{State}]$ are valid, negative assertions like $\neg(q:\texttt{State})$ are not valid as conditions in membership equational logic (because every membership axiom is positive, only positive conclusions can be reached).

There are some ways to overcome or mitigate this limitation. The language Maude includes built-in, extra-logical features to check whether \emph{or not} a term has a sort, and that solves the problem in practice. Without leaving the strictness of rewriting logic, it is not uncommon that the cases where a property must be left undefined are identifiable and expressible as valid conditions. This opens some possibilities. One of them is to create for each state $q$ for which the property $p : \texttt{State} \rightarrow [s]$ is undefined as many states as values are in the universe of the sort $s$. For example, if $s$ is \verb!Bool!, we would replace $q$ by $q_\texttt{true}$ and $q_\texttt{false}$. We have to duplicate (or, in general, replicate) all rules and equations needed so that each of the new states has the same behaviour as the old one had. Then, we define $p(q_b)=b$, making $p$ total. States at the other system can always synchronise with the appropriate $q_b$. This procedure can result in very complicated specifications if applied literally for each partially defined property. In our experience, undefinedness is not often needed, so we hope this is not such a serious constraint.

We define the synchronous composition of two topmost plain rewrite systems, $\calr_1$ and $\calr_2$, with respect to a synchronisation criteria $Y$ that uses only totally defined properties. It is denoted by $\calr_1\|_Y\calr_2$ and is defined to be a new plain rewrite system $\calr=(S,{\le},\Sigma,E,M,R)$ that happens to be topmost as well. The set of rules $R$ is described below. The rest of the elements of $\calr$ are defined as the almost disjoint (see below) union of the respective components of $\calr_1$ and $\calr_2$ (that is, $S=S_{\calr_1}\union S_{\calr_2}$, and so on), except for the following:
\begin{itemize}
\item There is in $S$ a new sort $\texttt{State}$ $(=\texttt{State}_\calr$) and a constructor $\langle\_,\_\rangle : \texttt{State}_{\calr_1} \times \texttt{State}_{\calr_2} \rightarrow [\texttt{State}]$.

\item There is a new constant \verb!init! of sort \verb"State" and an equation $\texttt{init} = \langle\texttt{init}_{\calr_1},\texttt{init}_{\calr_2}\rangle$.

\item For each $(p_1,p_2)\in Y$, with $p_i : \texttt{State}_{\calr_i} \rightarrow [s_i]$, we assume, as we did in Section~\ref{ers}, and for the same reasons, the existence of a common equational theory of $\text{lub}(s_1,s_2)$, subtheory of both $\calr_1$ and $\calr_2$. (This is where the unions are not disjoint.)

\item We add the following membership axiom, formalizing compatibility of states: 
\[\langle q_1,q_2\rangle : \texttt{State} \rif q_1:\texttt{State}_{\calr_1}\\
     \;\land\; q_2:\texttt{State}_{\calr_2}\\
     \;\land\; \bigwedge_{(p_1,p_2)\in Y} p_1(q_1)=p_2(q_2).\]

\item For each property $p : \texttt{State}_{\calr_i} \rightarrow [s_i]$ defined in $\calr_i$, for $i=1$ or $2$, there is in $\Sigma$ a declaration $p : \texttt{State} \rightarrow [s_i]$ and in $E$ an equation \emph{exporting} the property: ``$p(\langle q_1,q_2\rangle)=p(q_i) \rif \langle q_1,q_2\rangle : \texttt{State}$''.
\end{itemize}

The set of rules $R$ is built in the following way. For each pair of rules ``$q_1\rightarrow q'_1 \rif C_1$'' from $R_1$ and ``$q_2\rightarrow q'_2 \rif C_2$'' from $R_2$, there are three rules in $R$, namely:
\begin{itemize}
\item $\langle q_1,q_2\rangle\rightarrow\langle q'_1,q'_2\rangle \rif C_1 \land C_2 \land \langle q_1,q_2\rangle:\texttt{State} \land \langle q'_1,q'_2\rangle:\texttt{State}$,
\item $\langle q_1,q_2\rangle\rightarrow\langle q'_1,q_2\rangle \rif C_1 \land \langle q_1,q_2\rangle:\texttt{State} \land \langle q'_1,q_2\rangle:\texttt{State}$,
\item $\langle q_1,q_2\rangle\rightarrow\langle q_1,q'_2\rangle \rif C_2 \land \langle q_1,q_2\rangle:\texttt{State} \land \langle q_1,q'_2\rangle:\texttt{State}$.
\end{itemize}

We remind the reader that we assume each system to be its own namespace (with the exception of the common equational subtheories mentioned in Section~\ref{ers}). It is in this sense that we speak about (almost) disjoint unions in previous lines. In particular, when we join together equations coming from $\calr_1$ and from $\calr_2$, their effects do not overlap.

As in previous sections, some different ways to compose plain rewrite systems are intuitively seen to be equivalent. We denote by $\Sigma|_\texttt{State}$ the subset of the operators in $\Sigma$ that take a single argument of sort $\texttt{State}$, and by $T_{\Sigma/E,\texttt{State}}$, as usual, the set of terms of sort $\texttt{State}$ modulo equations.

\begin{defn}
Given two plain rewrite systems $\calr_i\in\textsf{\upshape RwSys}$, for $i=1,2$, we say that they are \emph{equivalent}, and denote it by $\calr_1\equiv\calr_2$, if there exist two bijections, $f:T_{\Sigma_1,\texttt{State}} \rightarrow T_{\Sigma_2,\texttt{State}}$ and $F:\Sigma_1|_\texttt{State} \rightarrow \Sigma_2|_\texttt{State}$, such that, for all $q_1,q'_1\in T_{\Sigma_1,\texttt{State}}$ and all $p_1\in\Sigma_i|_\texttt{State}$:
\begin{itemize}
\item the rewrite relation is preserved: $q_1\rightarrow_{\calr_1} q'_1 \Leftrightarrow f(q_1)\rightarrow_{\calr_2} f(q'_1)$;
\item result sorts for properties are the same: if $p_1$ is declared in $\calr_1$ as $p_1 : \texttt{State} \rightarrow s$, for a sort $s\in S_1$, then also $s\in S_2$ and $F(p_1)$ is declared in $\calr_2$ as $F(p_1) : \texttt{State} \rightarrow s$, and $T_{\Sigma_1/E_1,s} = T_{\Sigma_2/E_2,s}$;
\item values of properties are preserved: $p_1(q_1)=F(p_1)(f(q_1))$;
\item initial states are preserved: $f(\texttt{init})=\texttt{init}$.
\end{itemize} 
\end{defn}

This definition closely resembles Definition~\ref{ts-equiv}. As in previous sections, we allow expressions like ${\|_Y}_{i=1}^n\calr_i$, or $\|_{Y,I}\calr_i$ for some index set $I$, or just $\|_Y\calr_i$ if indexes are implied. Also, we denote by $\langle q_1,\dots,q_n\rangle$ a typical element of $\|_Y\calr_i$.

Although the operator ``$\|$'' for \textsf{RwSys} reflects the one for \textsf{EgRwSys}, they are of a different nature: $\calr_1\|_Y\calr_2$ is an irreducible expression in \textsf{EgRwSys}, but can be reduced in \textsf{RwSys} to a single rewrite system as already described.

In the next section we show how plain rewrite systems are useful as a more standard equivalent to egalitarian ones. Plain systems also have the advantage of being always single systems, even when they are the result of a composition. If one does not care about egalitarianism, plain rewrite systems can be taken as a basis for discussing non-egalitarian compositional specification.

\subsubsection{The split}
\label{split-ers}

The \emph{split} is a translation of rewrite systems: $\splitop : \textsf{\upshape EgRwSys}\rightarrow\textsf{\upshape RwSys}$. Translating an atomic system into a plain one is straightforward: given $\calr=(S,{\le},\Sigma,E,M,R)\in\textsf{\upshape atEgRwSys}$, its split is $\splitop(\calr)=(S',{\le},\Sigma,E,M,R')\in\textsf{RwSys}$, where $S'$ is the result of renaming in $S$ the sort \verb"State" to \verb"State'", and \verb"Stage" to \verb"State" (with the only aim of getting the top sort still being called \verb"State"), and $R'$ is the result of splitting each rule ``$t \trans{\ell} t' \rif C$'' in $R$ to produce the two rules ``$t \rightarrow \ell \rif C$'' and ``$\ell \rightarrow t' \rif C$'' in $R'$. (This rule splitting is, by the way, the reason for choosing the name \emph{split} for all similar operations in this paper.)

For a non-atomic system $\calr_1\|_Y\calr_2$, its \emph{split} is recursively defined by $\splitop(\calr_1\|_Y\calr_2) = \splitop(\calr_1)\|_Y\splitop(\calr_2)$. With this definition, it can be proved inductively (similar to Section~\ref{tr-split}) that the sets of properties for $\calr$ and for $\splitop(\calr)$ are the same for any $\calr\in\textsf{\upshape EgRwSys}$ and, thus, that the same relation $Y$ can be used in both sides of the definition. But the composition of plain systems in the right-hand side of this equation has to satisfy the two provisos: that the properties in $Y$ are totally defined and that the operands are topmost. Topmostness, in particular, deserves a few lines. We have not defined when an atomic egalitarian rewrite system is topmost, but the definition for plain systems can be readily adapted: when rewrites are only possible for complete state terms and only through complete transition terms. With this definition, it is easily seen that $\calr\in\textsf{\upshape atEgRwSys}$ is topmost iff $\splitop(\calr)\in\textsf{\upshape RwSys}$ is. Also easy is that the composition of topmost plain rewrite systems is topmost. Thus, repeated application of ``$\|$'' is allowed, and we are entitled to write $\splitop(\|_{Y,I}\calr_i)=\|_{Y,I}\splitop(\calr_i)$.

This definition of the $\splitop$ operation is theoretically sound, but can produce a large number of rules in the resulting system, so it is not suited for an eventual implementation. This is illustrated with the example at the end of Section~\ref{simple-example}.

As we observed in Section~\ref{tr-split}, the states and rewrites in $\splitop(\calr)$ can be taken as global states and global rewrites for the egalitarian system $\calr$.

\begin{prop}
For $\calr_1,\calr_2\in\textsf{\upshape EgRwSys}$, we have $\calr_1\equiv\calr_2 \Rightarrow \splitop(\calr_1)\equiv\splitop(\calr_2)$.
\end{prop}

\begin{proof}
As $\calr_1\equiv\calr_2$ implies $\mathrm{atomic}(\calr_1)=\mathrm{atomic}(\calr_2)$, either $\calr_1$ and $\calr_2$ are both atomic or none is. In the first case, indeed, $\calr_1=\calr_2$, and the result is immediate. In the second case, the result follows easily from the definitions. For example, the terms of sort $\texttt{State}$ in $\splitop(\calr_1)$ and those in $\splitop(\calr_2)$ are in bijection because both are built by applying tuple operators to the same set of \texttt{State} sorts from atomic components in different ways, and restricted by the same membership axioms, stemming from $\mathrm{criteria}(\calr_1)=\mathrm{criteria}(\calr_2)$.
\end{proof}

\subsection{Semantics}
\label{sem}

\begin{defn}
\label{sem-aers}
Given an atomic egalitarian rewrite system $\calr=(S,{\le},\Sigma,E,M,R)\in\textsf{\upshape atEgRwSys}$, we define its associated atomic egalitarian transition structure $\semop(\calr)=(Q,T,\rightarrow,P,g^0)\in\textsf{\upshape atEgTrStr}$ by:
\begin{itemize}
\item $Q$ is the set of ground terms of sort \verb"State" in $\calr$ modulo $E$, that is, $T_{\Sigma/E,\texttt{State}}$;
\item $T$ is the set of ground terms of sort \verb"Trans" in $\calr$ modulo $E$, that is, $T_{\Sigma/E,\texttt{Trans}}$;
\item $\rightarrow$ is the half-rewrite relation $\rightarrow^{eg}_\calr$ induced by $R$ according to Definition~\ref{defn:half-rewrite};
\item $P$ is $\Sigma|_\texttt{State}$ (the set of operators in $\Sigma$ whose domain is \verb"Stage" or one of its subsorts);
\item $g^0=[\texttt{init}]_{E}$.
\end{itemize}
\end{defn}

\begin{defn}
\label{sem-ers}
Given a non-atomic egalitarian rewrite system $\calr_1\|_Y\calr_2\in\textsf{\upshape EgRwSys}$, we define its associated non-atomic egalitarian transition structure by $\semop(\calr_1\|_Y\calr_2)=\semop(\calr_1)\|_Y\semop(\calr_2)\in\textsf{\upshape EgTrStr}$.
\end{defn}
Note that the same $Y$ can be used on both sides.

\begin{defn}
\label{sem-rs}
Given a plain rewrite system $\calr=(S,{\le},\Sigma,E,M,R)\in\textsf{\upshape RwSys}$, we define its associated plain transition structure $\semop(\calr)=(Q,\rightarrow,P,g^0)\in\textsf{\upshape TrStr}$ by:
\begin{itemize}
\item $Q$ is the set of ground terms of sort \verb"State" in $\calr$ modulo $E$, that is, $T_{\Sigma/E,\texttt{State}}$;
\item $\rightarrow$ is the rewrite relation $\rightarrow_\calr$ induced by $R$ according to Section~\ref{rs};
\item $P$ is $\Sigma|_\texttt{State}$;
\item $g^0=[\texttt{init}]_{E}$.
\end{itemize}
\end{defn}

\begin{prop}
\label{prop-sem-rs}
For plain rewrite systems $\calr_1,\calr_2\in\textsf{\upshape RwSys}$, each of them topmost with respective sets of totally defined properties $P_1$ and $P_2$, and for synchronisation criteria $Y \subseteq P_1\times P_2$, we have $\semop(\calr_1 \|_Y \calr_2) \equiv \semop(\calr_1) \|_Y \semop(\calr_2)$ (for the equivalence relation on \textsf{\upshape TrStr} from Definition~\ref{ts-equiv}).
\end{prop}

\begin{proof}
Easy application of the definitions. For example, the set of states for $\semop(\calr_1) \|_Y \semop(\calr_2)$ is the product $T_{\Sigma_1/E_1,\texttt{State}} \times T_{\Sigma_2/E_2,\texttt{State}}$ (that is, pairs of terms of sort $\texttt{State}$, one from each system) restricted by the conditions from $Y$. Meanwhile, the set of states for $\semop(\calr_1 \|_Y \calr_2)$ is $T_{\Sigma/E,\texttt{State}}$, where $\texttt{State}$ is defined by the constructor $\langle\_,\_\rangle : \texttt{State}_{\calr_1} \times \texttt{State}_{\calr_2} \rightarrow [\texttt{State}]$, and subject to the membership axioms drawn from $Y$. These sets can be easily put in bijection as required by Definition~\ref{ts-equiv}.
\end{proof}

As a way to validate our definitions of semantics, we state the following results.

\begin{prop}
For $\calr_1,\calr_2\in\textsf{\upshape EgRwSys}$, we have $\calr_1\equiv\calr_2 \Rightarrow \semop(\calr_1)\equiv\semop(\calr_2)$.
\end{prop}

\begin{proof}
Trivial for both the atomic case and the composed, inductive case.
\end{proof}

\begin{prop}
For $\calr_1,\calr_2\in\textsf{\upshape RwSys}$, we have $\calr_1\equiv\calr_2 \Rightarrow \semop(\calr_1)\equiv\semop(\calr_2)$.
\end{prop}

\begin{proof}
Easy thanks to the very similar-looking definitions of the two equivalence relations.
\end{proof}

\begin{thm}
\label{split-sem}
For every egalitarian rewrite system $\calr\in\textsf{\upshape EgRwSys}$ with totally defined properties and all whose components are topmost we have $\semop(\splitop(\calr)) \equiv \splitop(\semop(\calr))$ (for the equivalence relation on \textsf{\upshape TrStr} from Definition~\ref{ts-equiv}).
\end{thm}

This states the commutativity of one of the faces of the polyhedron in the introduction to Section~\ref{formal-definitions}. Indeed, this was the only face whose commutativity remained to be addressed.

\begin{proof}
We use structural induction on the shape of $\calr$. The base case is the atomic one. Let $\calr=(S,{\le},\Sigma,E,M,R)\in\textsf{\upshape atEgRwSys}$. From the definitions, it is not difficult to see that both $\splitop(\semop(\calr))$ and $\semop(\splitop(\calr))$ result in the plain transition structure
\[(T_{\Sigma/E,\texttt{Stage}}, {\rightarrow_\calr^\text{eg}}, \Sigma|_{\texttt{Stage}}, [\texttt{init}]_E).\]

For the composed case, $\calr=\calr_1\|_Y\calr_2$:
\begin{align*}
&\quad \semop(\splitop(\calr_1\|_Y\calr_2))\\
&= \qquad\textrm{(definition in Section~\ref{split-ers},
                  $\splitop$ for $\textsf{\upshape EgRwSys}$)}\\
&\quad \semop(\splitop(\calr_1)\|_Y\splitop(\calr_2))\\
&\equiv \qquad\textrm{(Proposition~\ref{prop-sem-rs})}\\
&\quad \semop(\splitop(\calr_1))\|_Y\semop(\splitop(\calr_2))\\
&= \qquad\textrm{(induction hypothesis)}\\
&\quad \splitop(\semop(\calr_1))\|_Y\splitop(\semop(\calr_2))\\
&= \qquad\textrm{(Definition~\ref{split-ets},
                  $\splitop$ for $\textsf{\upshape EgTrStr}$)}\\
&\quad \splitop(\semop(\calr_1)\|_Y\semop(\calr_2))\\
&= \qquad\textrm{(Definition~\ref{sem-ers},
                  $\semop$ for $\textsf{\upshape EgRwSys}$)}\\
&\quad \splitop(\semop(\calr_1\|_Y\calr_2)).
\end{align*}
\end{proof}

\subsection{The power of simple synchronisation}
\label{power}

Our mechanism for handling properties and synchronisation is quite simple: synchronisation is given by mere equality of properties, which are inherited untouched from a component to its composition. In contrast, other formalisms (some of which are mentioned in Section~\ref{related}) have complex channels for communication. And, in object-oriented technologies, for example, the interface for a container object may be quite different from the interfaces of the objects it contains or uses. Our concept of connector is that of a simple wire. A complex \emph{connector}, with some logic in it, has to be modelled as a new, intermediate component.

We expect that this choice for utmost simple interaction will make the compositional analysis of systems easier. The implicit danger is that such simplicity restricts the variety of systems that can be modelled. We show in this section two examples of how our simplicity is not a restriction.

\paragraph{Complex properties}
\label{complex}

In this example there are two components, $\cals_1$ and $\cals_2$, each providing a natural number by means of a property, $n_1$ and $n_2$, respectively. A different system $\cals'$ has a property $n'$, and we want something to the effect of $(\cals_1 \|_\emptyset \cals_2) \|_{\{(n_1+n_2,n')\}} \cals'$. In words: at each moment, only stages $g_1$, $g_2$, and $g'$ satisfying $n_1(g_1)+n_2(g_2)=n'(g')$ can be visited simultaneously.

This can be achieved by adding a third component able to perform addition. This component $\cals_+$ has to present three properties: $m_1$, $m_2$, and $s$, and has to guarantee that $s(g)=m_1(g)+m_2(g)$ at each stage $g$. For example, the system $\cals_+$ can be implemented with pairs of numbers as states and a single rule ``$(p,q) \trans{\ell} (p',q')$''. The properties are defined as $m_1((p,q))=p$, $m_2((p,q))=q$, and $s((p,q))=p+q$. The system $S_+$ has no restrictions about which state to visit next, but appropriate synchronisation makes the state of $S_+$ reflect the current values of $n_1$ and $n_2$ from the components. The composition we were looking for is $((\cals_1 \|_\emptyset \cals_2) \|_{\{(m_1,n_1),(m_2,n_2)\}} \cals_+) \|_{\{(s,n')\}} \cals'$. Or, if we prefer, $\|_{\{(m_1,n_1), (m_2,n_2), (s,n')\}} \{\cals_1, \cals_2, \cals_+, \cals'\}$.

\paragraph{Synchronising on relations other than equality}
\label{other-rel}

Suppose that we need to synchronise two systems not just by equality of the values of some properties, but on some other relation between them. For example, both properties are numbers and the relation is $<$. Or a property is a number, the other a set, and the relation is $\in$. This can be obtained by introducing an intermediate component.

More formally, suppose we have a system $\cals_1$ that defines a property $p_1$ with codomain $C_1$, and another $\cals_2$ that defines $p_2$ with codomain $C_2$, and we need to synchronise so that the values of the properties are related by a given relation $R\subseteq C_1\times C_2$. That is, something to the effect of $\cals_1 \|_{R(p_1,p_2)} \cals_2$. The intermediate component $\cals_R$ has a single rule ``$(p,q) \trans{\ell} (p',q') \rif R(p',q')$'', and properties $\pi_1((p,q))=p$ and $\pi_2((p,q))=q$. The composed system is $\|_{\{(p_1,\pi_1), (p_2,\pi_2)\}} \{\cals_1, \cals_2, \cals_R\}$.

\subsection{Executability}

Executability is one big benefit of rewriting logic as a system specification formalism. A specification that satisfies some requirements (see below) can be executed and model checked. This is in addition to being formal and, thus, amenable to formal analysis. If our proposal for synchronous composition is to provide some value, executability of composed systems is essential.

In general, executability of each atomic component in isolation may be neither necessary nor sufficient to ensure executability of the composed system. Not necessary, because a specification that is not executable by itself, may become so when restricted or enriched by interaction with the appropriate environment. Not sufficient, because also synchronization criteria have to be evaluated at runtime.

The rewriting relation between terms (that is, whether a term $t$ can be rewritten to $t'$ in a rewrite system) is undecidable for arbitrary sets of equations and rules. In the setting of standard rewriting logic, in~\cite{Meseguer2008}, for instance, conditions are stated on how to make that relation effectively decidable. A system that satisfies those conditions is called a \emph{computable system}. The three conditions, stated here in a rather simplistic way, are:
\begin{itemize}
\item equality modulo a set $A$ of equational axioms (like commutativity or associativity of certain operators) is decidable;
\item the set of equations is ground Church-Rosser and ground terminating (modulo the axioms $A$);
\item rules are ground coherent with respect to the equations (modulo the axioms $A$).
\end{itemize}
These conditions are easy to meet. Usually, the rewrite systems a sensible programmer would code are computable.

We can adapt these conditions to atomic egalitarian rewrite systems. Only the third one deals with rules and, thus, only it needs to be adapted. In the standard setting, coherence means that if a rewrite is possible from a state term $t$ to another $t'$, then from any term in the equational class of $t$ a rewrite is possible to a term in the equational class of $t'$. This allows an execution engine to work by, first, reducing the current state term to its normal form with respect to the equations and, then, rewriting from the normal form; see~\cite[Sect. 6.3]{ClavelDELMMT2007} for complete explanations. In our setting, transition terms can be reduced by equations, in the same way as state terms can. Thus, we need coherence in the two phases:
\begin{itemize}
\item if a transition with term $\ell$ can be fired from a state with term $t$, then from any state term in the equational class of $t$ a transition can be fired whose term is in the equational class of $\ell$;
\item if a transition with term $\ell$ can reach a state with term $t'$, then from any transition term in the equational class of $\ell$ a state can be reached whose term is in the equational class of $t'$.
\end{itemize}
These two complementary conditions are required for a system in \textsf{\upshape atEgRwSys} to be considered \emph{computable}.

For a non-atomic egalitarian rewrite system, we also need that the synchronization criteria be effectively checkable. But these are equalities in what we have called in Section~\ref{ers} the common equational theory for the sort of the values involved in the synchronization. As this common theory is included in each atomic component that needs it, its decidability is already implied by the conditions above. Thus, we can call a non-atomic rewrite system \emph{computable} if all its atomic components are. With this definition, a system is computable iff its split is.

That a system is computable in the sense above is a basic requirement towards executability. The practical implementation of an execution engine or a model checker is to be expected to need additional requirements. In the case of Maude's standard engine and model checker, there are a series of conditions to ensure that any variable in the right-hand side or the conditions of a rule can be instantiated at runtime (maybe in a non-deterministic way). These are called \emph{admissibility conditions} in~\cite{ClavelDELMMT2007}.

More recently, narrowing-based procedures have been developed for symbolic execution and model checking of Maude specifications; see~\cite{MeseguerT2007} for details. These procedures do not require the same admissibility conditions, but have their own limitations. For example, they cannot handle conditional rules.

In Section~\ref{other-rel}, we proposed using the rule ``$(p,q) \trans{\ell} (p',q') \rif R(p',q')$''. This involves fresh, unbound variables that Maude's execution engine is not able to handle (that is, even when split into two standard rules). And it is conditional, so the existing narrowing engine cannot be used either. We discuss some possible ways out in the rest of this section.

The language Maude allows \emph{matching conditions} in equations and in rules. They take the form \verb!if pattern := term!. Semantically, this is equivalent to the equational condition \verb!if pattern = term!, but it allows Maude's engine to instantiate variables appearing in the \verb!pattern! by matching against the reduced \verb!term!; see~\cite{ClavelDELMMT2007} for precise descriptions of what a pattern is and related explanations. In particular, matching conditions can be used to make Maude's engine choose a value non-deterministically from a set:
\begin{verbatim}
op FiniteSetOfAdmissibleValues : -> Set{Value} .
eq FiniteSetOfAdmissibleValues = ...
vars V V' : Value .
var Rest : Set{Value} .
crl V => V' if V' Rest := FiniteSetOfAdmissibleValues .
\end{verbatim}
Assuming the empty-syntax set constructor used in \verb!V' Rest! is commutative and associative, the value of \verb!V'! is chosen from \verb!FiniteSetOfAdmissibleValues! non-deterministically.

A model of a memory cell, for one example, must be ready to store any value \emph{within a range}, corresponding to the computer's word size. It is to be expected that models of real-world systems are often restricted in a similar way to a finite set of different states. Incorporating this knowledge into the model using matching conditions may turn the system into an executable one.

This idea is not always helpful, however. The system $\cals_+$ from Section~\ref{complex}, for example, must be able to provide a \emph{sum property} for whatever values it receives, and it must be usable in any environment. So, an \emph{a priori} finite range should not be used. But systems like this are not meant to be executed in isolation. In the composition $\|_{\dots}\{\cals_1, \cals_2, \cals_+, \cals'\}$, if we know that the values coming from $\cals_1$ and $\cals_2$ are within some finite range, that will give us a handle to the executability of the whole, composed system. Syntactically, it seems that a good way to represent this in code is extending the syntax for synchronization criteria to allow \emph{matching criteria}. Instead of writing $\|_{\{(m_1,n_1), (m_2,n_2), (s,n')\}} \{\cals_1, \cals_2, \cals_+, \cals'\}$, as we did above, we would write something like $\|_{\{m_1:=n_1,\, m_2:=n_2,\, n':=s\}} \{\cals_1, \cals_2, \cals_+, \cals'\}$. We can even expect this criteria to be translated into matching conditions by the split operation.  All this, however, concerns the implementation, that is not available as yet.

In addition to making the composed system executable, this idea of injecting knowledge from one component into a neighbouring one is also related to the assume-guarantee technique for compositional verification (also mentioned in Section~\ref{comp-verif}). That is, the assumption that \emph{input} values satisfy some restrictions, either in the form of temporal formulas or otherwise, can make a system amenable to both execution and verification.

\section{A simple but complete example}
\label{simple-example}

We show the specification of an example system, small but completely developed, with the aim of showing the way to think and use our tools for specification. We use the language Maude extended with some syntactic constructs for synchronous composition that we intend to implement in the near future.

A few declarations of sorts and operators are used often. So as to avoid repeating them, we assume they are included in a common module that is implicitly and silently imported whenever needed.
\begin{verbatim}
--- In the common module
   sorts State Trans Stage .
   subsorts State Trans < Stage .
   op init : -> Stage .
\end{verbatim}

Some changes are convenient in the implementation with respect to the theoretical description. Properties will not be functions in our Maude extension. Property sorts are created with this parameterised module:
\begin{verbatim}
fmod PPTY{X :: TRIV} is
   sort Ppty{X} .
   op _@_ : Ppty{X} Stage -> [X$Elt] .
endfm
\end{verbatim}

In this module, \verb"X" is the name of the formal parameter, and \verb"TRIV" is the name of a \emph{trivial} standard theory in Maude, that just requires the existence of a sort \verb!Elt! in the actual parameter. If we need, for example, to declare a Boolean-valued property \verb"B", we import \verb!PPTY{Bool}! and declare \verb!op B : -> Ppty{Bool}!. The instantiated module \verb!PPTY{Bool}! makes available the sort \verb!Ppty{Bool}! and also the infix evaluation operator \verb!_@_ : Ppty{Bool} Stage -> [Bool] .! Thus, to get the value of property \verb"B" at the stage \verb"G", we write the expression \verb"B @ G" (instead of \verb"B(G)"), which has kind \verb![Bool]!, in this case.

This example is inspired in one from~\cite{DeNicola1993}. Two trains move independently on the same linear railway divided in track sections. They both move rightwards, one track section at a time. The train on the left is not allowed to reach the one on the right, so that, when the trains are in consecutive track sections, the right one has to move.

\begin{center}
\begin{tikzpicture}[auto]
\draw [thick] (0,0) -- (8,0);
\draw [thick] (0,0.4) -- (8,0.4);
\draw [fill] (0.95,-0.1) rectangle (1.05,0.5);
\draw [fill] (2.95,-0.1) rectangle (3.05,0.5);
\draw [fill] (4.95,-0.1) rectangle (5.05,0.5);
\draw [fill] (6.95,-0.1) rectangle (7.05,0.5);
\node (train) at (2,0.53) {\scalebox{-1}[1]{\includegraphics[width=15mm]{train.pdf}}};
\node (train) at (4,0.53) {\scalebox{-1}[1]{\includegraphics[width=15mm]{train.pdf}}};
\draw [|->|,dashed] (4,-0.3) -- (6,-0.3);
\end{tikzpicture}
\end{center}

The two trains are modelled by identical modules. We leave the control of the movements to an external module, so that the model of each train needs not care about the existence of another train. Moreover, the behaviour of a train is likely to be the same no matter which track section it is sitting on, so we prefer that sections are not part of the model of each train, but of a different system that we call the \emph{reckoner}.

We structure the complete system in two levels. First, the two trains, \verb!LTRAIN! and \verb!RTRAIN!, and the \verb"RECKONER" are combined into a module called \verb"RECKONED-TRAINS"; then, in the second level, this system is combined with the \verb"CONTROLLER". Graphically:
\begin{center}
\begin{tikzpicture}[auto]
\draw (-1.5, -0.8) rectangle (4.5, 2.2);

\node [box] (ltrain) at (0, 1) {\verb"LTRAIN"};
\node [box] (rtrain) at (0, 0) {\verb"RTRAIN"};
\node [box] (reckoner) at (3, 0.5) {\verb"RECKONER"};
\node [box] (controller) at (7, 0.5) {\verb"CONTROLLER "};

\node (title) at (0.1, 1.9) {\verb"RECKONED-TRAINS"};

\draw [densely dotted, out=east, in=north] (ltrain) edge[)-(] (reckoner.north);
\draw [densely dotted, out=east, in=south] (rtrain) edge[)-(] (reckoner.south);
\draw [densely dotted, out=east, in=west] (reckoner.10) edge[)-] (4.5, 0.6);
\draw [densely dotted, out=east, in=west] (reckoner.-10) edge[)-] (4.5, 0.4);
\draw [densely dotted, out=west, in=south] (4.5, 0.2) edge[-(] (reckoner.south);
\draw [densely dotted, out=west, in=east] (controller.170) edge[)-(] (4.5, 0.6);
\draw [densely dotted, out=west, in=east] (controller.180) edge[)-(] (4.5, 0.4);
\draw [densely dotted, out=west, in=east] (controller.190) edge[)-(] (4.5, 0.2);
\end{tikzpicture}
\end{center}
Each system is represented as a box. Rounded connectors on the edges represent the properties each system provides. Dotted-line wires represent synchronisation. The reckoner, for example, needs to synchronise with the trains to know when each is moving. The controller sees the system \verb"RECKONED-TRAINS" as a black box that provides three properties to synchronise on. Inside the system \verb"RECKONED-TRAINS", we see that its three properties are originally defined as properties of the reckoner. This is another difference between the implementation and the theoretical description: although the internals of the \verb"RECKONED-TRAINS" module are not hidden, we prefer to pretend they are, and the \verb"CONTROLLER" behaves as if they were.

We proceed in a top-down fashion to translate the diagram into Maude. This is the top-level system:
\begin{verbatim}
mod CONTROLLED-TRAINS is
   pr RECKONED-TRAINS || CONTROLLER
   sync on RECKONED-TRAINS.areConsec = CONTROLLER.areConsec
        /\ RECKONED-TRAINS.isSomeMoving = CONTROLLER.doMove
        /\ RECKONED-TRAINS.isRMoving = CONTROLLER.doMoveR .
endm
\end{verbatim}
The keyword \verb"pr" is short for \emph{protecting}, a way to import modules in Maude. Synchronisation criteria are specified after the keywords \verb"sync on" (which is not standard Maude). They correspond to the set we called $Y$ above, and to dotted lines in the picture. The names of the properties are prefixed with the names of their modules, so as to avoid ambiguities. The six properties used in the synchronisation criteria (three from each system) are defined below.

This is the three-way composition that produces the system \verb"RECKONED-TRAINS", again assuming some properties to be defined below:
\begin{verbatim}
--- In module RECKONED-TRAINS
   pr LTRAIN || RTRAIN || RECKONER
   sync on LTRAIN.isMoving = RECKONER.isLMoving
        /\ RTRAIN.isMoving = RECKONER.isRMoving .
\end{verbatim}
This has to be complemented by the declaration and definition of the properties that have been used in the criteria for \verb"CONTROLLED-TRAINS":
\begin{verbatim}
--- Also in module RECKONED-TRAINS
   pr PPTY{Bool} .
   ops areConsec isSomeMoving isRMoving : -> Ppty{Bool} .
   var G : Stage .
   eq areConsec @ G = areConsec @ RECKONER(G) .
   eq isSomeMoving @ G = isSomeMoving @ RECKONER(G) .
   eq isRMoving @ G = isRMoving @ RECKONER(G) .
\end{verbatim}
These properties are defined for \verb"RECKONED-TRAINS" based on the properties of the \verb"RECKONER" component, following the design in the figure above. Instead of using the properties automatically inherited from the components, for clarity, we prefer to define explicitly new properties for the composed system, even if they are called the same. The design of our extension for the language Maude includes the existence of projection operators for the stage term with the same name as the components. Thus, \verb!areConsec @ RECKONER(G)! is the value of the property \verb!areConsec! at the \verb"RECKONER" part of the stage \verb"G". (We have argued before that the concept of global state is better avoided for distributed systems, so we would better interpret the use of the variable \verb"G" as convenient syntax.)

We have arrived at the level of the atomic components. The systems that model the two trains are identical, with a single state called \verb"stopped", a single transition called \verb"moving", and a single property called \verb!isMoving!, used to \emph{inform} the reckoner. The egalitarian rule is written with the transition term in the middle.
\begin{verbatim}
--- In modules LTRAIN and RTRAIN
   op stopped : -> State .
   op moving : -> Trans .
   rl stopped =[ moving ]=> stopped .

   pr PPTY{Bool} .
   op isMoving : -> Ppty{Bool} .
   eq isMoving @ moving = true .
   eq isMoving @ stopped = false .
\end{verbatim}
The property is defined as true at transitions and false at states. This may seem natural for a property called \verb!isMoving!, but it is not necessary. If the specification were more refined, for instance so as to allow a train to \emph{do something} while staying on the same track section, there would be a transition corresponding to that, and \verb!isMoving! would be false at that transition.

As for the \verb"RECKONER", it keeps in its state the distance (number of track sections) between the trains---an integer number. It has three rules, corresponding to the movement of each train and of both at the same time:
\begin{verbatim}
--- In module RECKONER
   subsort Int < State .
   ops lmoving|_ rmoving|_ 2moving|_ : Int -> Trans .
   var D : Int .
   rl D =[ lmoving | D ]=> D - 1 .
   rl D =[ rmoving | D ]=> D + 1 .
   rl D =[ 2moving | D ]=> D .
\end{verbatim}
The declaration \verb!subsort Int < State! is a way of saying that an integer is a state term by itself. Transitions with the first rule allow the distance to become zero, and even negative. This is just a reckoner, so it has nothing to say about it. The controller has to take care.

This is one instance in which the transitions have to \emph{remember} some value from the state, \verb"D" in this case; otherwise, each \verb"D" could be interpreted as a different variable (see Section~\ref{aers}). This is sometimes counter-intuitive. It is probably more intuitive to view \verb"D" as the \emph{context} in which the movement takes place, instead of being a parameter of the movement. We have chosen a syntax for transition terms that enforces this view: \verb!lmoving | D!.

The four properties of the \verb"RECKONER" are defined like this:
\begin{verbatim}
--- Also in module RECKONER
   pr PPTY{Bool} .
   ops areConsec isSomeMoving isLMoving isRMoving : -> Ppty{Bool} .
   var T : Trans .
   eq areConsec @ 1 = true .
   eq areConsec @ D = false [otherwise] .
   eq isSomeMoving @ D = false .
   eq isSomeMoving @ T = true .
   eq isLMoving @ (lmoving | D) = true .
   eq isLMoving @ (rmoving | D) = false .
   eq isLMoving @ (2moving | D) = true .
   eq isLMoving @ D = false .
   eq isRMoving @ (lmoving | D) = false .
   eq isRMoving @ (rmoving | D) = true .
   eq isRMoving @ (2moving | D) = true .
   eq isRMoving @ D = false .
\end{verbatim}
We do not define the value of \verb!areConsec! while the trains are moving. We assume that the controller (the ultimate user of this property) only needs to know its value when the trains are stopped. However, this is part of the interface (of \verb"RECKONER" and, then, of \verb"RECKONED-TRAINS") and we must be careful and consistent with our assumptions throughout the specification.

Only the controller is left. Its task is to ensure that, when the trains are in consecutive track sections, the right one moves, alone or otherwise. When the trains are not in consecutive sections, any one can move, or both. As shown above, we synchronise on three criteria: the first for the controller to be aware of consecutive-train situations; the second so that the controller can command some train to move; the third so that the controller can command the right train to move.

The word \emph{command} must be correctly understood. Synchronisation works in a symmetrical way: any system can be seen as commanding the other. Intuitively, in this case, the controller will set its property \verb"doMove" to true in some situations, and this will make mandatory for the trains to execute some action for which \verb!isSomeMoving! is true. This is what we dub as the controller commanding the trains.

It often happens that controllers have complete meanings by themselves, and can be applied to different base systems. This is the case with the one we are describing. In an abstract way, the controller's task is to detect when the base system is in a particular situation and command that only a particular action be allowed: the right train must move when the left one is next to it; only deposits are allowed in a bank account whose balance is zero; defensive moves are advised when our king is in trouble. Seen in this way, the names of the properties for the controller would be better chosen agnostic: instead of \verb!areConsec!, use \verb!isMarkedState!; instead of \verb!doMoveR!, use \verb!doMarkedAction!; instead of \verb!doMove!, use \verb!doAnyAction!. For the time being, however, we keep using train-related names.

The controller has two states: one, \verb!consec!, to represent trains in consecutive track sections; the other, \verb!nonConsec!, for the rest. A different transition is needed from each state, so that the one leaving from \verb!consec! commands (through synchronisation) a movement of the right train. No more refinement is needed.
\begin{verbatim}
--- In module CONTROLLER
   ops consec nonConsec : -> State .
   ops fromConsec fromNonConsec : -> Trans .
   rl consec =[ fromConsec ]=> consec .
   rl consec =[ fromConsec ]=> nonConsec .
   rl nonConsec =[ fromNonConsec ]=> consec .
   rl nonConsec =[ fromNonConsec ]=> nonConsec .
\end{verbatim}
Note that the first two rules define the same transition, as do the last two.

Finally, this is the definition of the properties for the controller that we have already used to synchronise with \verb"RECKONED-TRAINS":
\begin{verbatim}
--- Also in module CONTROLLER
   pr PPTY{Bool} .
   ops areConsec doMove doMoveR : -> Ppty{Bool} .
   var S : State .
   var T : Trans .
   eq areConsec @ consec = true .
   eq areConsec @ nonConsec = false .
   eq doMove @ S = false .
   eq doMove @ T = true .
   eq doMoveR @ fromConsec = true .
\end{verbatim}
When moving away from a non-consecutive state, any movement of the trains is valid. Thus, \verb"doMoveR" must be left undefined at \verb"fromNonConsec": setting it to false would prevent the movement of the right train.

This finishes the specification of the whole system. Let us describe informally how it works at a low level, that is, looking at the internals of the specifications. Requiring that \verb"doMove" equals \verb"isSomeMoving" (as we do in the synchronisation criteria for \verb"CONTROLLED-TRAINS") has the effect that states in \verb"RECKONER" can only be simultaneous with states in \verb"CONTROLLER", and the same for transitions. The equality of the two properties called \verb"areConsec" ensures, then, that the state \verb"consec" in \verb"CONTROLLER" is always visited simultaneously with the state \verb"1" in \verb"RECKONER". The only possible transition from \verb"consec" is called \verb"fromConsec" and sets the property \verb"doMoveR" to true. The synchronisation criterion \verb"doMoveR = isRMoving", then, mandates that a transition at which \verb"isRMoving" is true be executed. There are two such transitions in \verb"RECKONER": \verb"lmoving" and \verb"2moving". Each of them makes \verb"RTRAIN" execute its \verb"moving" transition, as can be inferred by continuing this chain of implications inside the definition of \verb"RECKONED-TRAINS" and its components.

Let us use this example to illustrate how compositional verification would work, even if only informally. The temporal property we are interested in proving is the absence of crashes between the trains: $\TBox\neg\texttt{crash}$, where the proposition \verb"crash" originates in the \verb"RECKONER" and can be defined as equivalent to $\texttt{D}<1$. However, the \verb"RECKONER" by itself cannot guarantee the formula. Its satisfaction depends on the \verb"CONTROLLER", who, in turn, does not know about distances among trains.

The system \verb"RECKONED-TRAINS" can guarantee $\TBox\neg\texttt{crash}$ if it assumes something from the environment, namely, the formula:
\[\varphi=\texttt{areConsec} \rightarrow (\neg\texttt{isSomeMoving} \TU \texttt{isRMoving}).\]
(We are using our Boolean-valued properties as Boolean propositions.) This is something that the \verb"CONTROLLER" is ready to guarantee if we replace the names of the Boolean-valued properties by the equivalent ones at its side:
\[\varphi=\texttt{areConsec} \rightarrow (\neg\texttt{doMove} \TU \texttt{doMoveR}).\]
The synchronisation criteria allow us to see these two formulas as the same one. In conclusion, we are saying that the \verb"CONTROLLER" guarantees $\varphi$, the \verb"RECKONED-TRAINS" guarantee $\TBox\neg\texttt{crash}$ under the assumption of $\varphi$, and from that we can assert $\texttt{CONTROLLED-TRAINS}\models\TBox\neg\texttt{crash}$. Adapting assume-guarantee techniques to formally justify this is left for future work.

As a last service from this example, we use it to illustrate the split operation. It allows to transform all the above into an equivalent, monolithic specification, that can be fed to Maude's tools, in case it would be needed. We illustrate how it works for the smaller, but still composed system \verb"RECKONED-TRAINS". The states of the split system are triples of component stages:
\begin{verbatim}
op <_,_,_> : LTRAIN.Stage RTRAIN.Stage RECKONER.Stage -> [State] .
\end{verbatim}
The membership axiom asserting when a triple can be considered a true state is this:
\begin{verbatim}
var L : LTRAIN.Stage .
var R : RTRAIN.Stage .
var K : RECKONER.Stage .
cmb < L, R, K > : State if isMoving @ L = isLMoving @ K
                        /\ isMoving @ R = isRMoving @ K .
\end{verbatim}
The condition is a direct translation of the synchronisation criteria.

For atomic components, the split of a rule produces two. Thus, the split of each train system produces two rules:
\begin{verbatim}
rl stopped => moving .
rl moving => stopped .
\end{verbatim}
The split of the reckoner produces six. This makes a total of $2\times2\times6=24$ combinations of rules. For each of these 24 triplets of rules, one has to consider that all the three rules can execute in one step, or just one or two of them. Each of these possibilities produces a rule in the resulting split system. This gives a large number of rules for a quite simple system. This is theoretically sound, but certainly not the way to go for an implementation. In some cases, like the present one, some membership conditions turn out to be trivially false, and the corresponding rules can be removed. Consider, for instance, this composed rule:
\begin{verbatim}
crl < moving, stopped, lmoving | D > => < moving, moving, lmoving | D >
   if < moving, stopped, lmoving | D > : State /\
   /\ < moving, moving, lmoving | D > : State .
\end{verbatim}
The first membership condition is always true and the second is always false, no matter what the value of \verb"D" is. A straightforward static analysis should be able to detect these cases and to avoid generating the rules.

In contrast, in the rule
\begin{verbatim}
crl < moving, stopped, lmoving | D > => < stopped, stopped, D - 1 >
   if < moving, stopped, lmoving | D > : State /\
   /\ < stopped, stopped, D - 1 > : State .
\end{verbatim}
both conditions are always true, so that it must be generated, but a static analysis should be able to remove the condition from it. The implementation and all the issues related to it will be the subject of future work.

\section{Related work}
\label{related}

The literature on compositionality is vast. From it we highlight and discuss below a few items: theoretical foundations, general frameworks, language and system implementations. Although this section is big, inevitably, many interesting models, languages, and systems are left out. We hope that the items that have made it into our discussion provide a useful overview of the available landscape.

Before delving into the main topic of compositionality, we review the precedents of what we have called \emph{egalitarianism}.

\paragraph{Egalitarianism}

Several temporal logics have been proposed that make joint use of actions and propositions on states: ACTL*~\cite{DeNicola1990}, RLTL~\cite{Sanchez2012}, SE-LTL~\cite{Chaki2004}, TLR*~\cite{Meseguer2008}. There are also transition structures with mixed ingredients: LKS~\cite{Chaki2004}, L2TS~\cite{DeNicola1995}. Although all of them bring actions (or transitions) to the focus, none tries to be utterly egalitarian, as we do.

The best move towards egalitarianism we know of is the temporal logic of rewriting, TLR* (which was an inspiration for the present work). The explanations and examples in~\cite{Meseguer2008} are good arguments for an egalitarian view. Consider this formula to express fairness in the execution of a rule with label $\ell$: $\TBox\TDiam\mathtt{enabled}.\ell \rightarrow \TBox\TDiam\mathtt{taking}.\ell$. The proposition $\mathtt{enabled}.\ell$ is on states: it means that the current state of the system admits the rule $\ell$ to be applied to it. But $\mathtt{taking}.\ell$ is on transitions: it means that a transition is being executed according to rule $\ell$. The simplicity of the formula is only possible by being egalitarian.

Plain TLR*, as described in~\cite{Meseguer2008}, stays a step away from our goal, because transitions are given by proof terms, that univocally determine one origin state and one destination state for each transition. TLR* uses proof-term patterns (called \emph{spatial actions}), that are used literally in temporal formulas. The problem is that, in this way, a TLR* formula is tied to a particular algebraic specification (one in which the pattern makes sense). In contrast, an LTL or CTL formula is meaningful by itself and can be used on any system specification by using atomic proposition definitions as interfaces. Notably, propositions on transitions have been added to plain TLR*, in some way or another, in all the implementations of model checkers for (the linear-time subset of) TLR* that we are aware of \cite{BaeM10,BaeM12,BaeM15,MartinVM14}. None of them, however, tries to allow a same proposition to be defined both in states and in transitions, which we need for flexible synchronization.

\paragraph{Our previous work}

Our paper~\cite{Martin2016b} contains a first definition of the synchronous composition of rewrite systems. There, we proposed to synchronise the execution of rules from different systems based on the coincidence of (atomic) rule labels. This reflects the synchronisation of actions in process algebras and in automata, for example. We also proposed to synchronise states by agreement on the Boolean values of propositions defined on them. We implemented that concept of synchronisation on Maude. That proposal had the advantage that it used standard machinery already existing in Maude: rule labels are basic elements of Maude's syntax, and propositions are customarily defined and used to build LTL formulas to be used with Maude's model checker. Why is the present, much more involved paper needed? We refer the reader to Section~\ref{choices}. In short: Boolean-valued propositions are not enough to allow flexible synchronisation and value-passing; we need to give more substance to transitions; we want to be able to synchronise an action at one system with several consecutive ones at the other system. A complex realistic example like the one on the alternating bit protocol in~\cite{MartinVM2018} would not be possible in our previous setting.

In a different topic, the paper~\cite{MartinVM2018} also describes the use of parameterised programming to add encapsulation to our setting. We have already mentioned it in Section~\ref{composed-ppty}. We outline it roughly refering to the example from Section~\ref{simple-example}, on two controlled trains. First, a so-called theory is used to state that a train is any system that defines a Boolean-valued property called \verb"isMoving". Requirements for reckoners and controllers are similarly stated. These are our interface specifications. The composition is specified in a parameterised module, whose formal parameters are the theories (that is, the interfaces). Thus, the composition can only be specified using the formal names and the properties in the interfaces. The particular implementations of trains and the other components are written and the needed properties are defined. Finally, the parameters of the composition module are instantiated with the component implementations, producing an implementation of the complete system.

\paragraph{Process algebras}

The parallel composition operator in process algebras (like CCS and CSP) is a notable precedent of our synchronous composition. Classic references are~\cite{Hoare1978,Milner1980}. A well-known example in CCS involves the specification of a coffee machine \verb"CM" and a computer scientist \verb"CS":
\begin{align*}
&\texttt{CM} := \texttt{coin}\,.\,\overline{\texttt{coffee}}\,.\,\texttt{CM},\\
&\texttt{CS} := \overline{\texttt{coin}}\,.\,\texttt{coffee}\,.\,\overline{\raisebox{2mm}{}\texttt{paper}}\,.\,\texttt{CS}.
\end{align*}
They are composed requiring actions \verb"coin" and \verb"coffee" to execute simultaneously:
\[(\texttt{CM}\,|\,\texttt{CS})\setminus\texttt{coin}\setminus\texttt{coffee}.\]

This can be translated to rewriting logic by coding two modules, \texttt{CM} and \texttt{CS}, with simple rules representing the actions, trivially defining properties, and then:
\begin{verbatim}
mod ComposedSystem is
   pr CM || CS
   sync on CM.isAcceptingCoin = CS.isInsertingCoin
        /\ CM.isProducingCoffee = CS.isTakingCoffee .
endm
\end{verbatim}

Process algebras were initially designed as theoretical tools. They focus on actions and synchronisation, and do not provide any means to specify internal computations, or to handle complex data types. However, later developments have taken process algebras as a basis for practical modelling and verification tools. Examples are occam~\cite{WelchB2004}, SCEL~\cite{DeNicolaLLLMMMPTV2015}, FSP+LTSA~\cite{MageeK2006}, CSP$\|$B~\cite{ButlerL2005}, and LOTOS and the CADP tool~\cite{GaravelLS17}.

\paragraph{Automata and labelled transition structures}

Both automata and labelled transition structures are formally digraphs, whose nodes represent states and whose edges represent actions. Classic references are~\cite{Hopcroft2006} for automata, and~\cite{Clarke1999} for labelled transition structures (although they are rather called \emph{Kripke structures} there). In automata, edges are labelled with atomic identifiers from some alphabet. In labelled transition structures, each node is assigned a set of propositions that are said to hold at that state; edge labels can be used as well. A labelled transition structure is thought of as modelling a system; an automaton is rather thought of as accepting a language.

These two graphs represent the coffee machine and the scientist from the previous section, and can be seen either as automata or as labelled transition structures:
\begin{center}
\begin{tikzpicture}[auto]
\node [state] (cm) at (0, 0) {\texttt{CM}};
\node [state] (cm') at (2.5, 0) {};
\draw [->, bend left] (cm) edge node {\texttt{coin}} (cm');
\draw [->, bend left] (cm') edge node {\texttt{coffee}} (cm);

\node [state] (cs) at (5, 0) {\texttt{CS}};
\node [state] (cs') at (7, 1) {};
\node [state] (cs'') at (7, -1) {};
\draw [->, bend left] (cs) edge node {\texttt{coin}} (cs');
\draw [->, bend left] (cs') edge node {\texttt{coffee}} (cs'');
\draw [->, bend left] (cs'') edge node {\texttt{paper}} (cs);
\end{tikzpicture}
\end{center}

We have used labelled transition structures as semantics in this paper. However, they can be used as modelling tools by themselves. Their limitation, that they share with process algebras, is that they do not provide means for specifying internal computations or handling complex data types. But variations and extensions abound, and different definitions for synchronous products have been proposed; see references cited above. Worth mentioning are I/O automata~\cite{LynchT1989}, explicitly designed with compositionality and distributed systems in mind.

\paragraph{Petri nets}

Petri nets do not include compositionality as a built-in ingredient. According to \cite{Sobocinski2016}, the first to address this problem was Mazurkiewicz in \cite{Mazurkiewicz1987}. Notably, while compositionality deserves no mention in \cite{Reisig1985} (one of the standard references on Petri nets), it is the subject of several chapters in the much more modern \cite{Jensen2009} (one of the standard references on coloured Petri nets), and is also introduced in~\cite{Reisig2013}.

Typically, there are two ways to compose Petri nets. One is given by hierarchical nets, that is, nets in which a transition can represent a complete separate net, that is described independently. The second way is to identify, or \emph{fuse}, either places or transitions from two different nets. For example, the coffee machine and the scientist can be modelled and then composed by fusing transitions like this:
\begin{center}
\begin{tikzpicture}[auto,scale=0.7]
\node [pnstate] (a1) at (-1, 3) {$\bullet$};
\node [pntrans] (coin1) at (-1, 2) [label=right:\texttt{coin}] {};
\node [pnstate] (b1) at (-1, 1) {};
\node [pntrans] (coffee1) at (-1, 0) [label=right:\texttt{coffee}] {};
\node (dummy1) at (-2, 1.7) {};
\draw [->] (a1) edge (coin1);
\draw [->] (coin1) edge (b1);
\draw [->] (b1) edge (coffee1);
\draw [->, out=270, in=270, looseness=2] (coffee1) to (dummy1) to [->, out=90, in=90, looseness=2] (a1);

\node at (1.5, 1) {\huge$\|$};

\node [pnstate] (a2) at (4, 3) {$\bullet$};
\node [pntrans] (coin2) at (4, 2) [label=left:\texttt{coin}] {};
\node [pnstate] (b2) at (4, 1) {};
\node [pntrans] (coffee2) at (4, 0) [label=left:\texttt{coffee}] {};
\node [pnstate] (c2) at (4, -1) {};
\node [pntrans] (paper2) at (4, -2) [label=left:\texttt{paper}] {};
\node (dummy2) at (5, 1.2) {};
\draw [->] (a2) edge (coin2);
\draw [->] (coin2) edge (b2);
\draw [->] (b2) edge (coffee2);
\draw [->] (coffee2) edge (c2);
\draw [->] (c2) edge (paper2);
\draw [->, out=270, in=270, looseness=2] (paper2) to (dummy2) to [->, out=90, in=90, looseness=2] (a2);

\node at (6, 1) {\huge$=$};

\node [pnstate] (a31) at (8, 3) {$\bullet$};
\node [pntrans] (coin3) at (9, 2) [label=right:\texttt{coin}] {};
\node [pnstate] (b31) at (8, 1) {};
\node [pntrans] (coffee3) at (9, 0) [label=right:\texttt{coffee}] {};
\node (dummy31) at (7, 1.7) {};
\draw [->] (a31) edge (coin3);
\draw [->] (coin3) edge (b31);
\draw [->] (b31) edge (coffee3);
\draw [->, out=270, in=270, looseness=2] (coffee3) to (dummy31) to [->, out=90, in=90, looseness=2] (a31);
\node [pnstate] (a32) at (10, 3) {$\bullet$};
\node [pnstate] (b32) at (10, 1) {};
\node [pnstate] (c32) at (10, -1) {};
\node [pntrans] (paper32) at (10, -2) [label=left:\texttt{paper}] {};
\node (dummy32) at (11.5, 1.2) {};
\draw [->] (a32) edge (coin3);
\draw [->] (coin3) edge (b32);
\draw [->] (b32) edge (coffee3);
\draw [->] (coffee3) edge (c32);
\draw [->] (c32) edge (paper32);
\draw [->, out=270, in=270, looseness=2] (paper32) to (dummy32) to [->, out=90, in=90, looseness=2] (a32);
\end{tikzpicture}
\end{center}
Some approaches propose the introduction of interfaces, that allow to see each component net as a black box. That is the case of the recent work described in \cite{BruniMMS13}.

Petri nets seem to have been extended in all possible directions. There are extensions in which complex data types are handled, and others that can include pieces of executable code within a net. However, for the proposals about compositionality that we are aware of, synchronisation is performed only on basic terms, and value-passing is not addressed.

\paragraph{Tile logic}

Tile logic was introduced in~\cite{GadducciM2000}, and is closely related to rewriting logic. In short, tile logic is rewriting logic with side effects for composition. A tile is written as $t\xrightarrow[b]{a}t'$ with $a$ being the condition for, and $b$ the effect of, rewriting $t$ to $t'$. The intuitive meaning of that tile is: ``the term $t$ is rewritten to the term $t'$, producing an effect $b$, but the rewrite can only happen if the variables of $t$ (that represent as yet unspecified subcomponents) are rewritten with a cumulative effect $a$.'' Effects are given by terms of any complexity.

Tiles can be composed in three ways: vertically, so that $t$ is first rewritten to $t'$ and then $t'$ to $t''$ according to another tile, if conditions are met; horizontally, so that several subterms of $t$ are rewritten according to some tiles, their effects being composed; and in parallel, representing the simultaneous but independent execution of several rewrites. Conditions and effects are the boundaries (interfaces) of the tiles, and they are the only thing important for composition.

Connections between tile logic and rewriting logic have been drawn in~\cite{MeseguerM97} and~\cite{BruniMM02}, mainly in the language of category theory.

\paragraph{$\textrm{Span}(\textrm{Graph})$}

The definition of $\textrm{Span}(\textrm{Graph})$ is based on category theory, although this is largely irrelevant for its use. It is described, for example, in~\cite{GianolaKS2017}. Each component is given by a graph (a kind of automaton), and each has left and right interfaces. Each action that takes place in the graph is reflected (or not) in the left and/or right interfaces.

There are three forms of composition available: parallelism without communication, parallelism with communication (where the right interface of a component gets attached to the left interface of another), and sequentiality.

The split of the interface in two parts, left and right, is an artifact only needed so that each component can be linearly assembled to other two, one in each side. This eases formalisation. There is no difference for an interface to be left or right. In particular, it is not the case that left is input and right is output. Indeed, there are no in and out interfaces, it depends on how they are used (it is the same in our setting).

Complex compositions can be written as algebraic expressions. For example, for processes (graphs) $A$ and $B$, a circuit with a feedback connection is written like this: $\eta * ((A * B) \otimes \iota) * \varepsilon$, and drawn like this:
\begin{center}
\begin{tikzpicture}[auto]
\node [state] (A) at (1.9, 0.8) {$A$};
\node [state] (B) at (3.1, 0.8) {$B$};
\draw (0, 0) edge node{$\eta$} (0.4, 0);
\draw (0.4, 0) edge (1, 0.8);
\draw (1, 0.8) edge (A);
\draw (A) edge (B);
\draw (B) edge (4, 0.8);
\draw (4, 0.8) edge (4.6, 0);
\draw (0.4, 0) edge (1, -0.8);
\draw (1, -0.8) edge node{$\iota$} (4, -0.8);
\draw (4, -0.8) edge (4.6, 0);
\draw (4.6, 0) edge node{$\varepsilon$} (5, 0);
\end{tikzpicture}
\end{center}
That is, $\eta$ is producing a fork with two copies of its input, $\varepsilon$ is doing the opposite at the other end, and $\iota$ is identity (a simple piece of wire). The operator $*$ represents sequence, and $\otimes$ is parallelism. The circuit has to be described from left to right, even though the bottom wire is intended to work the other way.

\paragraph{Asynchronous messages}

Asynchronous message passing can be used to achieve some extent of modular design. Actions that involve more than one component (because they implement interactions) are replaced by two or more: one or several actions in which a component produces and sends a message, and one or several actions in which the other component receives and deals with the message. The components must agree on the structure of the messages they exchange, which can be seen as the specification of the interface. The components and the messages are made to be elements of a set, which allows new elements, like messages, to be added or removed on the fly. This is a usual technique in Maude, where there is even an object-based notation, described in~\cite[Ch.~11]{ClavelDELMMT2007}, to ease the coding of components and messages and sets of them.

Asynchronous message passing has its limitations, though. In each component, the logic for message passing has to be mixed with the specification of its proper workings. This hinders reusability. In addition, some systems are better seen as synchronous, and are difficult to emulate with asynchronous messages. A controller, for example, would have a hard time trying to efficiently manage some component systems just by placing asynchronous messages. Also, asynchrony is not appropriate when physical components are physically geared: it is not satisfactory, nor realistic, that something that must be immediate is implemented by a multiple exchange of messages.

\paragraph{Behavioural programming}

Use cases and scenarios are techniques used in the requirement-collection phase of software development. Behavioural programming, as described for example in \cite{Harel2012}, proposes that it is possible to keep these up to the coding phase, transformed into synchronised threads, so called \emph{behavioural threads}, or just \emph{behaviours}, each one corresponding to a use-case. The concept of behaviour here is really elementary. For instance, when coding the rules for the game of tic-tac-toe, there is a behavioural thread that prevents marking cell $(1,1)$ if it is already marked; other eight behaviours do the same for the other eight cells; yet another enforces alternating turns; and so on. The paper \cite{Harel2012} hints at the possibility that, for complex systems, behavioural threads may need to be grouped into \emph{nodes}, each thread synchronised only with others in its node, and nodes synchronised by external events.

The synchronisation mechanism for behavioural threads is based on the \emph{request-wait-block} paradigm. Each thread, at each synchronisation point, \emph{requests} some events (needs some of them to go on), \emph{waits} for some others (wants to be informed if they happen), and \emph{blocks} some others. The set of events is global and shared. When all threads reach a synchronisation point, somehow an event is chosen that is requested by some thread and blocked by none, and so the system goes on to the next simultaneous synchronisation point. Separation of concerns (computation from interaction) is not perfect, as requesting, blocking, and waiting are an indispensable part of the logic of each thread.

Behavioural programming is claimed to allow incremental development, because new behaviours can be added to an already running system with no need to modify existing ones (as long as new events are not needed). Several implementations of behavioural programming are available---see at \texttt{www.wisdom.weizmann.ac.il/\textasciitilde bprogram}.

\paragraph{Aspect-oriented programming}

Aspects are concerns of a system that \emph{cross-cut} the system's base functionality. Typical examples of aspects are tracing, error handling, and monitoring. Aspect-oriented programming, as explained in \cite{Kiczales1997}, proposes, first, grouping in one module all code related to a given aspect, and, second, establishing at which points in the base code the aspect code must be executed. The language must provide some means for establishing those \emph{join points}. For instance, some monitoring method can be needed each time an object is created, or each time a method named \verb"imdangerous" is executed. This frees the programmer of the base code from worrying about such concerns.

The motivations for aspect-oriented programming are not too different from ours. Our properties can be seen as the mechanism to establish when other (aspect) code must be run. A difference is our insisting in simultaneity. Aspect code is usually executed right after or right before the need is found in the base code. For aspects, it is more difficult to prevent, and easier to react.

Aspects are seen as an ingredient \emph{added} to an existing programming paradigm, like object orientation. The most notable implementation of aspects is AspectJ, built on Java (and with a nice plug-in for Eclipse). Objects provide a first tool for modularity. In our case, all modularity is based on the synchronous composition operation.

A frequent criticism to aspects is that programmers who are coding base functionality do not know what is really happening in the whole system---aspect code executes out of their control. In principle, aspects can modify variables, create new objects, and so on. A monitor, for example, can be designed, not just to detect unsafe states, but to take control of the system and bring it back to safety. Such powerful aspects can be seen as breaking modularity. Our own view on this is that a monitor or controller does not introduce new behaviours on the system. If the system can be taken to a safe state, it is because it was already able to get there. Left to itself, the system would probably not choose that path, so the task of the controller is to enforce the path to safety and ban the rest.

\paragraph{Assembly theories and interface theories}

Assembly theories and interface theories share the goal of establishing requirements for something to be called an assembly or interface (which are similar concepts). They are described, respectively, in~\cite{HennickerKW2014} and~\cite{deAlfaroH2005}. Assemblies and interfaces, as characterized in those papers, are rich structures that exert a distributed control on the components by prescribing how and when they can or must communicate or collaborate.

Assembly theories are more abstract than interface theories. Where interface theories consider ports and connections, assembly theories just suppose that some form of communication between components is provided. Where interface theories use automata to specify requirements to a component's interface, assembly theories assume some requirements can be made. Each at its level of abstraction, both formalisms discuss refinement, encapsulation of several components into one, composition of assemblies/interfaces. And both state needed properties such as compositionality of refinement.

A singular feature of interface theories is that they see the environment in an  optimistic way, that is: a component is declared valid if there is \emph{some} environment in which it can be used. The component restricts the environment, which must know the correct way to \emph{use} the component. An interface specification needs not be ready for any inputs it could receive, but just for the correct ones, which makes the specification simpler. The authors claim this eases incremental design, although incremental verification requires the pessimistic view.

\paragraph{Coordination}

The goal of coordination is to make different components work together. The components may have been coded in different languages, reside in different servers, with different architectures. The paper~\cite{PapadopoulosA1998} is a comprehensive reference, though old.

There is a large variety of languages and systems for coordination. Some enforce complete separation of concerns between coordination and computation; others require each component to include its part of the coordination logic. Communication may be port to port, with well-defined interfaces; or there may be a central repository of data shared by all components; or anything in between. Channels may be active components, able to hold data and follow some logic; or they may be just lines joining two ports. Some models allow dynamic reconfiguration of the component network; others are more static. Communication can be synchronous or asynchronous.

Coordination is a very general term, and some of the proposals we have discussed above can be seen as belonging to it. Let us name a few additional examples. Linda, with all its variants and implementations, is one of the best known coordination languages. See~\cite{Wells2005} for a relatively recent take. It is based on the idea of a shared repository of data and relies on each component using coordination primitives in appropriate ways. REO is at the other extreme. It enforces separation of concerns and is based on composing basic channels to provide port to port communication between components. Basic channels can be of any imaginable sort (synchronous or not, lossy\dots) and the means to compose them are very flexible. REO is described in~\cite{Arbab2004}.

BIP stands for \emph{behaviour}, \emph{interaction}, \emph{priority}---the three \emph{layers} of a composed specification, as proposed by the authors. The behaviour of atomic components is specified by automata (of a special kind) some of whose actions are also taken as port names for communication. These automata are a specification of requirements on the component, whose real implementation can be made using any language or tool. Interaction is performed through connectors linking ports in potentially complex ways. Among the interactions that are allowed at any given time, the one with the highest priority is chosen and performed. Interaction and priority together implement control. The paper~\cite{BasuBS2008} has a good overview. Several implementations exist that allow to use the BIP framework within programming languages like Java and C++.

\section{Future work}
\label{future}

Bringing compositionality to rewriting logic in the sense discussed in this paper can open it to new fields of application like coordination models and component-based software development. Also, hardware specification could benefit. But all this is quite speculative right now. In this section we describe the more feasible tasks we intend to apply ourselves to in the near future.

In addition to those tasks, the ultimate test for our tool would be its use to model and analyse a real system, preferably not just a realistic one, but one drawn from the real world, either one in which rewriting logic or other formal methods have already been used with success, and that we can approach with our tools and compare results, or a new system that can particularly benefit from our methods. This will only be possible once the implementation is working and the other lines of work described below have been at least partially explored.

\subsection{Implementation}

We need a usable prototype of the synchronous composition operation that supports running all the examples in this paper. It will be developed by extending Full Maude. This is a reimplementation of the Maude interpreter in Maude itself. It is described, for instance, in~\cite{ClavelDELMMT2007}. Full Maude has the advantage of being easily extensible. It is the natural choice to implement our tool.

At the syntactic level, our implementation must be ready to accept, first, transition terms in rules, and, second, the synchronous composition operator \verb"||" with its \verb"sync on" clauses. Also, the split of the synchronous composition must be computed when needed so as to allow an easy use of Maude's execution and search engines, and the LTL model checker.

That plain implementation will need to be improved in two aspects. First, the number of rewrite rules of the resulting split system can easily get very large and inefficient. Depending on the case, some or many of the conditions of these rules are trivially false or trivially true, and a straightforward static analysis can remove them. Even whole rules can be statically removed sometimes. This can have a large impact on performance.

Second, we have found that coding the specifications for composed systems in the way we propose is not always easy nor intuitive; but we have also found that some tricks and shortcuts can be used. See, in particular, \cite{MartinVM2018}. We must consider including some of these tricks in the implementation to increase our tool's usability.

\subsection{Strategies}

As important as the possibility of assembling \emph{physical} components (like a processor and a memory) is that of using \emph{abstract} components to control others. By \emph{abstract} we mean that they do not represent physical entities. Several examples in this paper illustrate the use of components as controllers or---as they are usually called in rewriting logic, with roughly the same meaning---strategies. The basic idea is that we make mandatory that some actions in the base system are synchronised with some actions in the controller. Thus, if the controller refuses to execute a particular action, the corresponding one on the base system is prevented.

This allows us to specify a base system with all its non-deterministic capabilities as one component, and use it only under the control of another component; even in different ways under different controls. This is the idea used in~\cite{Lescanne1989,Clavel1997,Verdejo2012} to implement Knuth-Bendix-like completion as a basic set of correct rules on which different strategies are applied to get different actual procedures. Again, the same idea is used in~\cite{Bachmair2003} for congruence closure. And in~\cite{Marti2004,Eker2007} insertion sort is implemented as a base system with a single rule for swapping cell contents and, then, a control on how to use that rule.

The language Maude includes the nice possibility of working at the meta-level. That is, a Maude module can be \emph{meta-represented}, stored, handled, and used as an object within another (meta-level) Maude module. Rewriting can be performed in a controlled way at the meta-level. This is a direct way to implement strategies in Maude. However, working at the meta-level is sometimes cumbersome, and requires a deeper understanding of Maude. And the results obtained are difficult to export to other formalisms, even to other rewriting-based ones. That is why object-level strategy languages are desirable and have indeed been developed, in Maude and in other formalisms (for Maude, see~\cite{Marti-Oliet2009}).

There are two tasks to be addressed. The first, theoretical one is putting our proposal in the context of the existing work on strategies: studying which kinds of strategies are implementable using synchronous composition, with which advantages and drawbacks, seeing if we can help clarify the philosophical question on the nature of strategies. The second task is implementing translations that render strategies written in some strategy language as components in rewriting logic and synchronisation criteria to compose them.

\subsection{Compositional verification}
\label{comp-verif}

Compositional specification is very nice, but it gets even better if one can go on working compositionally, particularly for verification purposes. This observation has certainly been made many times in the past and there is abundant work on modular, or component-wise, verification. A well-known paradigm is \emph{assume-guarantee}, initially proposed in~\cite{Pnueli1985}. According to it, each component assumes some nice behaviour from the rest of the system and, under such an assumption, guarantees its own nice behaviour. If each component can be proven to satisfy some condition of this type, conclusions can be drawn on the whole composed system. We have already mentioned this at the end of Section~\ref{tr-split}.

We do not foresee substantial theoretical developments on our part here. We just intend to adapt existing relevant work to our framework. This can include the implementation of new commands or facilities for compositional model checking.

\section{Conclusion}
\label{conclusion}

We are confident that the most appealing parts of our work are still to come. Strategies, compositional verification, runtime verification, trace model checking, coordination models, component-based software development---compositionality can turn rewriting logic suitable for some of these or other fields. These explorations will be enjoyed, hopefully, at some future time. But all the necessary theoretical basis for them has been laid down in this paper.

We have explained why transitions and states must be treated as equals, and how this is possible in so-called \emph{egalitarian} rewrite systems, which allow for complex transition terms instead of the usual atomic labels. We have proposed a flexible means for specifying how several rewrite systems are synchronised, based on agreement on the values of \emph{properties}. We have shown the power of properties in several realistic examples, and we have also justified why we need all that power, even though it entails more complexity. We have developed the theory for transition structures as well, so that our rewrite systems get a semantic ground. We have described the \emph{split} operations, that translate sets of synchronised egalitarian rewrite systems into standard ones.

We can now begin our further explorations walking on firm ground.

\paragraph*{Acknowledgements} This paper owes a good share of whatever value it may have to the very careful and useful remarks and suggestions from its anonymous referees.

\bibliographystyle{acmtrans}
\bibliography{modspec}

\begin{thebibliography}{}

\bibitem[\protect\citeauthoryear{Arbab}{Arbab}{2004}]{Arbab2004}
{\sc Arbab, F.} 2004.
\newblock {REO}: a channel-based coordination model for component composition.
\newblock {\em Mathematical Structures in Computer Science\/}~{\em 14,\/}~3,
  329--366.

\bibitem[\protect\citeauthoryear{Bachmair, Tiwari, and Vigneron}{Bachmair
  et~al\mbox{.}}{2003}]{Bachmair2003}
{\sc Bachmair, L.}, {\sc Tiwari, A.}, {\sc and} {\sc Vigneron, L.} 2003.
\newblock {Abstract Congruence Closure}.
\newblock {\em Journal of Automated Reasoning\/}~{\em 31,\/}~2, 129--168.

\bibitem[\protect\citeauthoryear{Bae and Meseguer}{Bae and
  Meseguer}{2010}]{BaeM10}
{\sc Bae, K.} {\sc and} {\sc Meseguer, J.} 2010.
\newblock {The Linear Temporal Logic of Rewriting Maude Model Checker}.
\newblock In {\em Rewriting Logic and Its Applications - 8th International
  Workshop, {WRLA} 2010, Held as a Satellite Event of {ETAPS} 2010, Paphos,
  Cyprus, March 20-21, 2010, Revised Selected Papers}, {P.~C. {\"{O}}lveczky},
  Ed. Lecture Notes in Computer Science, vol. 6381. Springer, 208--225.

\bibitem[\protect\citeauthoryear{Bae and Meseguer}{Bae and
  Meseguer}{2012}]{BaeM12}
{\sc Bae, K.} {\sc and} {\sc Meseguer, J.} 2012.
\newblock {A Rewriting-Based Model Checker for the Linear Temporal Logic of
  Rewriting}.
\newblock {\em Electr. Notes Theor. Comput. Sci.\/}~{\em 290}, 19--36.

\bibitem[\protect\citeauthoryear{Bae and Meseguer}{Bae and
  Meseguer}{2015}]{BaeM15}
{\sc Bae, K.} {\sc and} {\sc Meseguer, J.} 2015.
\newblock {Model checking Linear Temporal Logic of Rewriting formulas under
  localized fairness}.
\newblock {\em Sci. Comput. Program.\/}~{\em 99}, 193--234.

\bibitem[\protect\citeauthoryear{Basu, Bozga, and Sifakis}{Basu
  et~al\mbox{.}}{2008}]{BasuBS2008}
{\sc Basu, A.}, {\sc Bozga, M.}, {\sc and} {\sc Sifakis, J.} 2008.
\newblock Modeling heterogeneous real-time components in {BIP}.
\newblock In {\em Perspectives Workshop: Model Engineering of Complex Systems
  (MECS), 10.08. - 13.08.2008}, {U.~A{\ss}mann}, {J.~B{\'{e}}zivin}, {R.~F.
  Paige}, {B.~Rumpe}, {and} {D.~C. Schmidt}, Eds. Dagstuhl Seminar Proceedings,
  vol. 08331. Schloss Dagstuhl - Leibniz-Zentrum f{\"{u}}r Informatik, Germany.

\bibitem[\protect\citeauthoryear{Boudol and Castellani}{Boudol and
  Castellani}{1988}]{Boudol1988}
{\sc Boudol, G.} {\sc and} {\sc Castellani, I.} 1988.
\newblock {A non-interleaving semantics for CCS based on proved transitions}.
\newblock {\em Fundamenta Informaticae\/}~{\em 11}, 433--452.

\bibitem[\protect\citeauthoryear{Bruni, Melgratti, Montanari, and
  Soboci{\'{n}}ski}{Bruni et~al\mbox{.}}{2013}]{BruniMMS13}
{\sc Bruni, R.}, {\sc Melgratti, H.~C.}, {\sc Montanari, U.}, {\sc and} {\sc
  Soboci{\'{n}}ski, P.} 2013.
\newblock Connector algebras for {C/E} and {P/T} nets' interactions.
\newblock {\em Logical Methods in Computer Science\/}~{\em 9,\/}~3.

\bibitem[\protect\citeauthoryear{Bruni, Meseguer, and Montanari}{Bruni
  et~al\mbox{.}}{2002}]{BruniMM02}
{\sc Bruni, R.}, {\sc Meseguer, J.}, {\sc and} {\sc Montanari, U.} 2002.
\newblock Tiling transactions in rewriting logic.
\newblock {\em Electronic Notes on Theoretical Computer Science\/}~{\em 71},
  90--109.

\bibitem[\protect\citeauthoryear{Bruns and Godefroid}{Bruns and
  Godefroid}{1999}]{Bruns1999}
{\sc Bruns, G.} {\sc and} {\sc Godefroid, P.} 1999.
\newblock Model checking partial state spaces with 3-valued temporal logics.
\newblock In {\em Computer Aided Verification: 11th International Conference,
  {CAV}'99}, {N.~Halbwachs} {and} {D.~Peled}, Eds. Lecture Notes in Computer
  Science, vol. 1633. Springer-Verlag, Trento, Italy, 274--287.

\bibitem[\protect\citeauthoryear{Bruns and Godefroid}{Bruns and
  Godefroid}{2000}]{Bruns2000}
{\sc Bruns, G.} {\sc and} {\sc Godefroid, P.} 2000.
\newblock {Generalized Model Checking: Reasoning about Partial State Spaces}.
\newblock In {\em CONCUR 2000---Concurrency Theory: 11th International
  Conference}, {C.~Palamidessi}, Ed. Lecture Notes in Computer Science, vol.
  1877. Springer, University Park, PA, USA, 168--182.

\bibitem[\protect\citeauthoryear{Butler and Leuschel}{Butler and
  Leuschel}{2005}]{ButlerL2005}
{\sc Butler, M.} {\sc and} {\sc Leuschel, M.} 2005.
\newblock Combining {CSP} and {B} for specification and property verification.
\newblock In {\em FM 2005: Formal Methods}, {J.~Fitzgerald}, {I.~J. Hayes},
  {and} {A.~Tarlecki}, Eds. Springer Berlin Heidelberg, Berlin, Heidelberg,
  221--236.

\bibitem[\protect\citeauthoryear{Chaki, Clarke, Ouaknine, Sharygina, and
  Sinha}{Chaki et~al\mbox{.}}{2004}]{Chaki2004}
{\sc Chaki, S.}, {\sc Clarke, E.~M.}, {\sc Ouaknine, J.}, {\sc Sharygina, N.},
  {\sc and} {\sc Sinha, N.} 2004.
\newblock State/event-based software model checking.
\newblock In {\em {IFM}}. Lecture Notes in Computer Science, vol. 2999.
  Springer, 128--147.

\bibitem[\protect\citeauthoryear{Clarke, Grumberg, and Peled}{Clarke
  et~al\mbox{.}}{1999}]{Clarke1999}
{\sc Clarke, E. M.~J.}, {\sc Grumberg, O.}, {\sc and} {\sc Peled, D.~A.} 1999.
\newblock {\em Model Checking}.
\newblock MIT Press, Cambridge, MA, USA.

\bibitem[\protect\citeauthoryear{Clavel, Dur{\'{a}}n, Eker, Lincoln,
  Mart{\'{i}}-Oliet, Meseguer, and Talcott}{Clavel
  et~al\mbox{.}}{2007}]{ClavelDELMMT2007}
{\sc Clavel, M.}, {\sc Dur{\'{a}}n, F.}, {\sc Eker, S.}, {\sc Lincoln, P.},
  {\sc Mart{\'{i}}-Oliet, N.}, {\sc Meseguer, J.}, {\sc and} {\sc Talcott,
  C.~L.} 2007.
\newblock {\em All About {Maude} - A High-Performance Logical Framework, How to
  Specify, Program and Verify Systems in Rewriting Logic}. Lecture Notes in
  Computer Science, vol. 4350.
\newblock Springer, Berlin, Heidelberg.

\bibitem[\protect\citeauthoryear{Clavel and Meseguer}{Clavel and
  Meseguer}{1997}]{Clavel1997}
{\sc Clavel, M.} {\sc and} {\sc Meseguer, J.} 1997.
\newblock Internal strategies in a reflective logic.
\newblock In {\em Proceedings of the {CADE}-14 Workshop on Strategies in
  Automated Deduction}, {B.~Gramlich} {and} {H.~Kirchner}, Eds. Springer,
  Townsville, Australia, 1--12.

\bibitem[\protect\citeauthoryear{de~Alfaro and Henzinger}{de~Alfaro and
  Henzinger}{2005}]{deAlfaroH2005}
{\sc de~Alfaro, L.} {\sc and} {\sc Henzinger, T.~A.} 2005.
\newblock Interface-based design.
\newblock In {\em Engineering Theories of Software Intensive Systems},
  {M.~Broy}, {J.~Gr{\"u}nbauer}, {D.~Harel}, {and} {T.~Hoare}, Eds. Springer
  Netherlands, Dordrecht, 83--104.

\bibitem[\protect\citeauthoryear{{De Nicola}, Fantechi, Gnesi, and Ristori}{{De
  Nicola} et~al\mbox{.}}{1993}]{DeNicola1993}
{\sc {De Nicola}, R.}, {\sc Fantechi, A.}, {\sc Gnesi, S.}, {\sc and} {\sc
  Ristori, G.} 1993.
\newblock {An action-based framework for veryfying logical and behavioural
  properties of concurrent systems}.
\newblock {\em Computer Networks and ISDN Systems\/}~{\em 25,\/}~7, 761--778.

\bibitem[\protect\citeauthoryear{De~Nicola, Latella, Lafuente, Loreti,
  Margheri, Massink, Morichetta, Pugliese, Tiezzi, and Vandin}{De~Nicola
  et~al\mbox{.}}{2015}]{DeNicolaLLLMMMPTV2015}
{\sc De~Nicola, R.}, {\sc Latella, D.}, {\sc Lafuente, A.~L.}, {\sc Loreti,
  M.}, {\sc Margheri, A.}, {\sc Massink, M.}, {\sc Morichetta, A.}, {\sc
  Pugliese, R.}, {\sc Tiezzi, F.}, {\sc and} {\sc Vandin, A.} 2015.
\newblock {\em The SCEL Language: Design, Implementation, Verification}.
\newblock Springer International Publishing, Cham, 3--71.

\bibitem[\protect\citeauthoryear{{De Nicola} and Vaandrager}{{De Nicola} and
  Vaandrager}{1995}]{DeNicola1995}
{\sc {De Nicola}, R.} {\sc and} {\sc Vaandrager, F.} 1995.
\newblock Three logics for branching bisimulation.
\newblock {\em Journal of the Association for Computing Machinery\/}~{\em
  42,\/}~2 (mar), 458--487.

\bibitem[\protect\citeauthoryear{{De Nicola} and Vaandrager}{{De Nicola} and
  Vaandrager}{1990}]{DeNicola1990}
{\sc {De Nicola}, R.} {\sc and} {\sc Vaandrager, F.~W.} 1990.
\newblock Action versus state based logics for transition systems.
\newblock In {\em Semantics of Systems of Concurrent Processes}. Lecture Notes
  in Computer Science, vol. 469. Springer, 407--419.

\bibitem[\protect\citeauthoryear{Eker, Mart{\'{i}}-Oliet, Meseguer, and
  Verdejo}{Eker et~al\mbox{.}}{2007}]{Eker2007}
{\sc Eker, S.}, {\sc Mart{\'{i}}-Oliet, N.}, {\sc Meseguer, J.}, {\sc and} {\sc
  Verdejo, A.} 2007.
\newblock Deduction, strategies, and rewriting.
\newblock In {\em Proceedings of the 6th International Workshop on Strategies
  in Automated Deduction ({STRATEGIES} 2006)}, {M.~Archer}, {T.~B. de~la Tour},
  {and} {C.~Mu{\~{n}}oz}, Eds. Electronic Notes in Theoretical Computer
  Science, vol. 174. Elsevier, Seattle, WA, USA, 3--25.

\bibitem[\protect\citeauthoryear{Gadducci and Montanari}{Gadducci and
  Montanari}{2000}]{GadducciM2000}
{\sc Gadducci, F.} {\sc and} {\sc Montanari, U.} 2000.
\newblock The tile model.
\newblock In {\em Proof, Language, and Interaction, Essays in Honour of Robin
  Milner}, {G.~D. Plotkin}, {C.~Stirling}, {and} {M.~Tofte}, Eds. The {MIT}
  Press, 133--166.

\bibitem[\protect\citeauthoryear{Garavel, Lang, and Serwe}{Garavel
  et~al\mbox{.}}{2017}]{GaravelLS17}
{\sc Garavel, H.}, {\sc Lang, F.}, {\sc and} {\sc Serwe, W.} 2017.
\newblock From {LOTOS} to {LNT}.
\newblock In {\em {ModelEd}, {TestEd}, {TrustEd} - Essays Dedicated to {Ed
  Brinksma} on the Occasion of His 60th Birthday}. 3--26.

\bibitem[\protect\citeauthoryear{Gianola, Kasangian, and Sabadini}{Gianola
  et~al\mbox{.}}{2017}]{GianolaKS2017}
{\sc Gianola, A.}, {\sc Kasangian, S.}, {\sc and} {\sc Sabadini, N.} 2017.
\newblock Cospan/span(graph): an algebra for open, reconfigurable automata
  networks.
\newblock In {\em 7th Conference on Algebra and Coalgebra in Computer Science,
  {CALCO} 2017, June 12-16, 2017, Ljubljana, Slovenia}, {F.~Bonchi} {and}
  {B.~K{\"{o}}nig}, Eds. LIPIcs, vol.~72. Schloss Dagstuhl - Leibniz-Zentrum
  fuer Informatik, 2:1--2:17.

\bibitem[\protect\citeauthoryear{Godefroid and Huth}{Godefroid and
  Huth}{2005}]{Godefroid2005}
{\sc Godefroid, P.} {\sc and} {\sc Huth, M.} 2005.
\newblock Model checking vs. generalized model checking: Semantic minimizations
  for temporal logics.
\newblock In {\em Proc. 20th Annual IEEE Symposium on Logic in Computer Science
  (LICS' 05)}. IEEE, Chicago, IL, USA, 158--167.

\bibitem[\protect\citeauthoryear{Harel, Marron, and Weiss}{Harel
  et~al\mbox{.}}{2012}]{Harel2012}
{\sc Harel, D.}, {\sc Marron, A.}, {\sc and} {\sc Weiss, G.} 2012.
\newblock {Behavioural Programming}.
\newblock {\em Communications of the Association for Computing
  Machinery\/}~{\em 55,\/}~7, 90--100.

\bibitem[\protect\citeauthoryear{Hennicker, Knapp, and Wirsing}{Hennicker
  et~al\mbox{.}}{2014}]{HennickerKW2014}
{\sc Hennicker, R.}, {\sc Knapp, A.}, {\sc and} {\sc Wirsing, M.} 2014.
\newblock Assembly theories for communication-safe component systems.
\newblock In {\em From Programs to Systems. The Systems perspective in
  Computing - {ETAPS} Workshop, {FPS} 2014, in Honor of Joseph Sifakis,
  Grenoble, France, April 6, 2014. Proceedings}, {S.~Bensalem}, {Y.~Lakhnech},
  {and} {A.~Legay}, Eds. Lecture Notes in Computer Science, vol. 8415.
  Springer, 145--160.

\bibitem[\protect\citeauthoryear{Hoare}{Hoare}{1978}]{Hoare1978}
{\sc Hoare, C. A.~R.} 1978.
\newblock {C}ommunicating {S}equential {P}rocesses.
\newblock {\em Communications of the Association for Computing
  Machinery\/}~{\em 21,\/}~8, 666--677.

\bibitem[\protect\citeauthoryear{Hopcroft, Motwani, and Ullman}{Hopcroft
  et~al\mbox{.}}{2006}]{Hopcroft2006}
{\sc Hopcroft, J.~E.}, {\sc Motwani, R.}, {\sc and} {\sc Ullman, J.~D.} 2006.
\newblock {\em {Introduction to Automata Theory, Languages, and Computation
  (3rd Edition)}}.
\newblock Addison-Wesley Longman Publishing Co., Inc., Boston, MA, USA.

\bibitem[\protect\citeauthoryear{Huth, Jagadeesan, and Schmidt}{Huth
  et~al\mbox{.}}{2001}]{Huth2001}
{\sc Huth, M.}, {\sc Jagadeesan, R.}, {\sc and} {\sc Schmidt, D.} 2001.
\newblock Modal transition systems: A foundation for three-valued program
  analysis.
\newblock In {\em Programming Languages and Systems: 10th European Symposium on
  Programming, {ESOP} 2001}, {D.~Sands}, Ed. Lecture Notes in Computer Science,
  vol. 2028. Springer, Genova, Italy, 155--169.

\bibitem[\protect\citeauthoryear{Jensen and Kristensen}{Jensen and
  Kristensen}{2009}]{Jensen2009}
{\sc Jensen, K.} {\sc and} {\sc Kristensen, L.~M.} 2009.
\newblock {\em Coloured {Petri} Nets}.
\newblock Springer Berlin Heidelberg, Berlin, Heidelberg.

\bibitem[\protect\citeauthoryear{Kiczales, Lamping, Mendhekar, Maeda, Lopes,
  Loingtier, Irwin, and Lopes}{Kiczales et~al\mbox{.}}{1997}]{Kiczales1997}
{\sc Kiczales, G.}, {\sc Lamping, J.}, {\sc Mendhekar, A.}, {\sc Maeda, C.},
  {\sc Lopes, C.~V.}, {\sc Loingtier, J.-M.}, {\sc Irwin, J.}, {\sc and} {\sc
  Lopes, C.} 1997.
\newblock {Aspect-Oriented Programming}.
\newblock In {\em ECOOP '97---Object-Oriented Programming}. Lecture Notes in
  Computer Science, vol. 1241. Springer-Verlag, Jyv{\"{a}}skyl{\"{a}}, Finland,
  220--242.

\bibitem[\protect\citeauthoryear{Kindler and Vesper}{Kindler and
  Vesper}{1998}]{Kindler1998}
{\sc Kindler, E.} {\sc and} {\sc Vesper, T.} 1998.
\newblock {ESTL: A temporal logic for events and states}.
\newblock In {\em Application and Theory of Petri Nets 1998: 19th International
  Conference, ICATPN '98}, {J.~Desel} {and} {M.~Silva}, Eds. Lecture Notes in
  Computer Science, vol. 1420. Springer, Lisbon, Portugal, 365--384.

\bibitem[\protect\citeauthoryear{Lescanne}{Lescanne}{1989}]{Lescanne1989}
{\sc Lescanne, P.} 1989.
\newblock {Completion Procedures as Transition Rules + Control}.
\newblock In {\em TAPSOFT '89: Proceedings of the International Joint
  Conference on Theory and Practice of Software Development}, {J.~D{\'{i}}az}
  {and} {F.~Orejas}, Eds. Lecture Notes in Computer Science, vol. 351.
  Springer, Berlin, Heidelberg, 28--41.

\bibitem[\protect\citeauthoryear{Lynch and Tuttle}{Lynch and
  Tuttle}{1989}]{LynchT1989}
{\sc Lynch, N.~A.} {\sc and} {\sc Tuttle, M.~R.} 1989.
\newblock An introduction to input/output automata.
\newblock {\em CWI Quarterly\/}~{\em 2}, 219--246.

\bibitem[\protect\citeauthoryear{Magee and Kramer}{Magee and
  Kramer}{2006}]{MageeK2006}
{\sc Magee, J.} {\sc and} {\sc Kramer, J.} 2006.
\newblock {\em Concurrency - state models and Java programs {(2.} ed.)}.
\newblock Wiley.

\bibitem[\protect\citeauthoryear{Mart{\'{i}}-Oliet, Meseguer, and
  Verdejo}{Mart{\'{i}}-Oliet et~al\mbox{.}}{2004}]{Marti2004}
{\sc Mart{\'{i}}-Oliet, N.}, {\sc Meseguer, J.}, {\sc and} {\sc Verdejo, A.}
  2004.
\newblock {Towards a strategy language for Maude}.
\newblock In {\em Proceedings of the Fifth International Workshop on Rewriting
  Logic and Its Applications (WRLA 2004)}, {N.~Mart{\'{i}}-Oliet}, Ed.
  Electronic Notes in Theoretical Computer Science, vol. 117. Elsevier,
  Barcelona, Spain, 417--441.

\bibitem[\protect\citeauthoryear{Mart{\'{i}}-Oliet, Meseguer, and
  Verdejo}{Mart{\'{i}}-Oliet et~al\mbox{.}}{2009}]{Marti-Oliet2009}
{\sc Mart{\'{i}}-Oliet, N.}, {\sc Meseguer, J.}, {\sc and} {\sc Verdejo, A.}
  2009.
\newblock {A Rewriting Semantics for Maude Strategies}.
\newblock In {\em Proceedings of the Seventh International Workshop on
  Rewriting Logic and its Applications (WRLA 2008)}, {G.~Rosu}, Ed. Electronic
  Notes in Theoretical Computer Science, vol. 238. Elsevier, Budapest, Hungary,
  227--247.

\bibitem[\protect\citeauthoryear{Mart{\'{\i}}n, Verdejo, and
  Mart{\'{\i}}{-}Oliet}{Mart{\'{\i}}n et~al\mbox{.}}{2014}]{MartinVM14}
{\sc Mart{\'{\i}}n, {\'{O}}.}, {\sc Verdejo, A.}, {\sc and} {\sc
  Mart{\'{\i}}{-}Oliet, N.} 2014.
\newblock {Model Checking TLR* Guarantee Formulas on Infinite Systems}.
\newblock In {\em Specification, Algebra, and Software - Essays Dedicated to
  Kokichi Futatsugi}, {S.~Iida}, {J.~Meseguer}, {and} {K.~Ogata}, Eds. Lecture
  Notes in Computer Science, vol. 8373. Springer, 129--150.

\bibitem[\protect\citeauthoryear{Mart{\'{i}}n, Verdejo, and
  Mart{\'{i}}-Oliet}{Mart{\'{i}}n et~al\mbox{.}}{2016a}]{Martin2016}
{\sc Mart{\'{i}}n, {\'{O}}.}, {\sc Verdejo, A.}, {\sc and} {\sc
  Mart{\'{i}}-Oliet, N.} 2016a.
\newblock {Egalitarian State-Transition Systems}.
\newblock In {\em Rewriting Logic and Its Applications: WRLA 2016},
  {D.~Lucanu}, Ed. Lecture Notes in Computer Science, vol. 9942. Springer,
  Eindhoven, The Netherlands, 98--117.

\bibitem[\protect\citeauthoryear{Mart{\'{i}}n, Verdejo, and
  Mart{\'{i}}-Oliet}{Mart{\'{i}}n et~al\mbox{.}}{2016b}]{Martin2016b}
{\sc Mart{\'{i}}n, {\'{O}}.}, {\sc Verdejo, A.}, {\sc and} {\sc
  Mart{\'{i}}-Oliet, N.} 2016b.
\newblock Synchronous products of rewrite systems.
\newblock In {\em Automated Technology for Verification and Analysis: {ATVA}
  2016}, {C.~Artho}, {A.~Legay}, {and} {D.~A. Peled}, Eds. Lecture Notes in
  Computer Science, vol. 9938. Springer, Cham, 141--156.

\bibitem[\protect\citeauthoryear{Mart{\'{\i}}n, Verdejo, and
  Mart{\'{\i}}{-}Oliet}{Mart{\'{\i}}n et~al\mbox{.}}{2018}]{MartinVM2018}
{\sc Mart{\'{\i}}n, {\'{O}}.}, {\sc Verdejo, A.}, {\sc and} {\sc
  Mart{\'{\i}}{-}Oliet, N.} 2018.
\newblock Parameterized programming for compositional system specification.
\newblock In {\em Rewriting Logic and Its Applications: WRLA 2018}, {V.~Rusu},
  Ed. Lecture Notes in Computer Science, vol. 11152. Springer.

\bibitem[\protect\citeauthoryear{Mazurkiewicz}{Mazurkiewicz}{1988}]{Mazurkiewicz1987}
{\sc Mazurkiewicz, A.} 1988.
\newblock {Compositional semantics of pure place/transition systems}.
\newblock In {\em Advances in Petri nets: APN 1987}, {G.~Rozenberg}, Ed.
  Lecture Notes in Computer Science, vol. 340. Springer, Oxford, UK, 307--330.

\bibitem[\protect\citeauthoryear{Meseguer}{Meseguer}{1992}]{Meseguer1992}
{\sc Meseguer, J.} 1992.
\newblock Conditional rewriting logic as a unified model of concurrency.
\newblock {\em Theoretical Computer Science\/}~{\em 96,\/}~1, 73--155.

\bibitem[\protect\citeauthoryear{Meseguer}{Meseguer}{2008}]{Meseguer2008}
{\sc Meseguer, J.} 2008.
\newblock {The Temporal Logic of Rewriting: A gentle introduction}.
\newblock In {\em Concurrency, Graphs and Models}, {P.~Degano}, {R.~D. Nicola},
  {and} {J.~Meseguer}, Eds. Lecture Notes in Computer Science, vol. 5065.
  Springer, Berlin, Heidelberg, 354--382.

\bibitem[\protect\citeauthoryear{Meseguer and Montanari}{Meseguer and
  Montanari}{1997}]{MeseguerM97}
{\sc Meseguer, J.} {\sc and} {\sc Montanari, U.} 1997.
\newblock Mapping tile logic into rewriting logic.
\newblock In {\em Recent Trends in Algebraic Development Techniques, 12th
  International Workshop, WADT'97, Tarquinia, Italy, June 1997, Selected
  Papers}. Lecture Notes in Computer Science, vol. 1376. Springer, 62--91.

\bibitem[\protect\citeauthoryear{Meseguer and Thati}{Meseguer and
  Thati}{2007}]{MeseguerT2007}
{\sc Meseguer, J.} {\sc and} {\sc Thati, P.} 2007.
\newblock Symbolic reachability analysis using narrowing and its application to
  verification of cryptographic protocols.
\newblock {\em Higher-Order and Symbolic Computation\/}~{\em 20}, 123--160.

\bibitem[\protect\citeauthoryear{Milner}{Milner}{1980}]{Milner1980}
{\sc Milner, R.} 1980.
\newblock {\em {A Calculus of Communicating Systems}}. Lecture Notes in
  Computer Science, vol.~92.
\newblock Springer-Verlag, Berlin, Heidelberg.

\bibitem[\protect\citeauthoryear{Papadopoulos and Arbab}{Papadopoulos and
  Arbab}{1998}]{PapadopoulosA1998}
{\sc Papadopoulos, G.~A.} {\sc and} {\sc Arbab, F.} 1998.
\newblock Coordination models and languages.
\newblock {\em Advances in Computers\/}~{\em 46}, 329--400.

\bibitem[\protect\citeauthoryear{Pnueli}{Pnueli}{1985}]{Pnueli1985}
{\sc Pnueli, A.} 1985.
\newblock In transition from global to modular temporal reasoning about
  programs.
\newblock In {\em Logics and Models of Concurrent Systems}, {K.~R. Apt}, Ed.
  Springer Berlin Heidelberg, Berlin, Heidelberg, 123--144.

\bibitem[\protect\citeauthoryear{Reisig}{Reisig}{1985}]{Reisig1985}
{\sc Reisig, W.} 1985.
\newblock {\em {Petri Nets: an Introduction}}.
\newblock EATCS Monographs in Theoretical Computer Science. Springer, Berlin,
  Heidelberg.

\bibitem[\protect\citeauthoryear{Reisig}{Reisig}{2013}]{Reisig2013}
{\sc Reisig, W.} 2013.
\newblock {\em Understanding Petri Nets: Modeling Techniques, Analysis Methods,
  Case Studies}.
\newblock Springer Berlin Heidelberg, Berlin, Heidelberg.

\bibitem[\protect\citeauthoryear{S{\'{a}}nchez and
  Samborski-Forlese}{S{\'{a}}nchez and Samborski-Forlese}{2012}]{Sanchez2012}
{\sc S{\'{a}}nchez, C.} {\sc and} {\sc Samborski-Forlese, J.} 2012.
\newblock {Efficient regular linear temporal logic using dualization and
  stratification}.
\newblock {\em Proceedings - 2012 19th International Symposium on Temporal
  Representation and Reasoning, TIME 2012\/}, 13--20.

\bibitem[\protect\citeauthoryear{Soboci{\'{n}}ski}{Soboci{\'{n}}ski}{2016}]{Sobocinski2016}
{\sc Soboci{\'{n}}ski, P.} 2016.
\newblock {Compositional model checking of concurrent systems, with Petri
  nets}.
\newblock In {\em Developements in Computational Models: DCM 2015 Proc.},
  {C.~A. Mu{\~{n}}oz} {and} {J.~A. P{\'{e}}rez}, Eds. Electronics Proceedings
  in Theoretical Computer Science, vol. 204. Open Publishing Association, Cali,
  Colombia, 19--30.

\bibitem[\protect\citeauthoryear{Verdejo and Mart{\'{i}}-Oliet}{Verdejo and
  Mart{\'{i}}-Oliet}{2012}]{Verdejo2012}
{\sc Verdejo, A.} {\sc and} {\sc Mart{\'{i}}-Oliet, N.} 2012.
\newblock {Basic completion strategies as another application of the Maude
  strategy language}.
\newblock In {\em Workshop on Reduction Strategies in Rewriting and Programming
  ({WRS2011})}, {S.~Escobar}, Ed. Electronic Proceedings in Theoretical
  Computer Science, vol.~82. Open Publishing Association, Novi Sad, Serbia,
  17--36.

\bibitem[\protect\citeauthoryear{Welch and Barnes}{Welch and
  Barnes}{2004}]{WelchB2004}
{\sc Welch, P.~H.} {\sc and} {\sc Barnes, F. R.~M.} 2004.
\newblock Communicating mobile processes.
\newblock In {\em Communicating Sequential Processes: The First 25 Years,
  Symposium on the Occasion of 25 Years of CSP, London, UK, July 7-8, 2004,
  Revised Invited Papers}, {A.~E. Abdallah}, {C.~B. Jones}, {and} {J.~W.
  Sanders}, Eds. Lecture Notes in Computer Science, vol. 3525. Springer,
  175--210.

\bibitem[\protect\citeauthoryear{Wells}{Wells}{2005}]{Wells2005}
{\sc Wells, G.} 2005.
\newblock Coordination languages: Back to the future with {Linda}.
\newblock In {\em Proceedings of WCAT'05}. 87--98.

\end{thebibliography}

\end{document}